\documentclass[journal]{IEEEtran}
\usepackage{amsmath}
\usepackage{amsfonts}
\usepackage{amssymb}
\usepackage{dsfont}
\usepackage{bm}
\usepackage{amsthm}
\usepackage{newlfont}
\usepackage{float}
\usepackage{hyperref}
\usepackage{pifont}%
\usepackage{enumerate,caption}
\usepackage{chngcntr}
\usepackage{mathtools}
\usepackage{breqn}
\usepackage{stmaryrd}
\usepackage{cite}
\usepackage{tikz}
\def\checkmark{\tikz\fill[scale=0.3](0,.35) -- (.25,0) -- (1,.7) -- (.25,.15) -- cycle;}
\newcommand{\cmark}{\checkmark}%
\newcommand{\xmark}{$\times$}%
\hypersetup{
    colorlinks=true, 
    linkcolor= blue, 
    citecolor=blue, 
    urlcolor= black
}

\usepackage{acronym}

\DeclareMathOperator*{\argmax}{arg\,max}

\acrodef{AoA}{Angle of Arrival}
\acrodef{AWGN}{Additive White Gaussian Noise}
\acrodef{BER}{Bit-Error-Rate}
\acrodef{BPSK}{Binary Phase-Shift Keying}
\acrodef{BSC}{Binary Symmetric Channel}
\acrodef{CDF}[CDF]{Cumulative Distribution Function}
\acrodef{CLT}[CLT]{Central Limit Theorem}
\acrodef{CSI}[CSI]{Channel State Information}
\acrodef{DMC}[DMC]{Discrete Memoryless Channel}
\acrodef{DMS}[DMS]{Discrete Memoryless Source}
\acrodef{iid}[i.i.d.]{independent and identically distributed}
\acrodef{LDPC}[LDPC]{Low-Density Parity-Check}
\acrodef{MAC}[MAC]{multiple-access channel}
\acrodef{MIMO}[MIMO]{Multiple-Input Multiple-Output}
\acrodef{MISO}{Multiple-Input Single-Output}
\acrodef{PDF}[PDF]{Probability Distribution Function}
\acrodef{PMF}[PMF]{Probability Mass Function}
\acrodef{PSD}{Power Spectral Density}
\acrodef{QPSK}{Quadrature Phase-Shift Keying}
\acrodef{SIMO}{Single-Input Multiple-Output}
\acrodef{SNR}{Signal-to-Noise Ratio}
\acrodef{wrt}[w.r.t.]{with respect to}
\acrodef{WSS}{Wide Sense Stationary}

\newcommand{\V}[1]{{{\mathbb{V}}\!\left(#1\right)}}

\newcommand{\avgH}[1]{{\mathbb{H}}\!\left(#1\right)}

\newcommand{\card}[1]{\ensuremath{\left|{#1}\right|}}           
\newcommand{\eqdef}{\ensuremath{\triangleq}}                    
\newcommand{\intseq}[2]{\ensuremath{\llbracket{#1},{#2}\rrbracket}}  

\renewcommand{\leq}{\leqslant}
\renewcommand{\geq}{\geqslant}

\DeclareMathAlphabet{\eurm}{U}{eur}{m}{n}
\DeclareMathAlphabet{\mathbsf}{OT1}{cmss}{bx}{n}
\DeclareMathAlphabet{\mathssf}{OT1}{cmss}{m}{sl}
\DeclareMathAlphabet{\mathcsf}{OT1}{cmss}{sbc}{n}

\DeclareSymbolFont{bsfletters}{OT1}{cmss}{bx}{n}  
\DeclareSymbolFont{ssfletters}{OT1}{cmss}{m}{n}
\DeclareMathSymbol{\bsfGamma}{0}{bsfletters}{'000}
\DeclareMathSymbol{\ssfGamma}{0}{ssfletters}{'000}
\DeclareMathSymbol{\bsfDelta}{0}{bsfletters}{'001}
\DeclareMathSymbol{\ssfDelta}{0}{ssfletters}{'001}
\DeclareMathSymbol{\bsfTheta}{0}{bsfletters}{'002}
\DeclareMathSymbol{\ssfTheta}{0}{ssfletters}{'002}
\DeclareMathSymbol{\bsfLambda}{0}{bsfletters}{'003}
\DeclareMathSymbol{\ssfLambda}{0}{ssfletters}{'003}
\DeclareMathSymbol{\bsfXi}{0}{bsfletters}{'004}
\DeclareMathSymbol{\ssfXi}{0}{ssfletters}{'004}
\DeclareMathSymbol{\bsfPi}{0}{bsfletters}{'005}
\DeclareMathSymbol{\ssfPi}{0}{ssfletters}{'005}
\DeclareMathSymbol{\bsfSigma}{0}{bsfletters}{'006}
\DeclareMathSymbol{\ssfSigma}{0}{ssfletters}{'006}
\DeclareMathSymbol{\bsfUpsilon}{0}{bsfletters}{'007}
\DeclareMathSymbol{\ssfUpsilon}{0}{ssfletters}{'007}
\DeclareMathSymbol{\bsfPhi}{0}{bsfletters}{'010}
\DeclareMathSymbol{\ssfPhi}{0}{ssfletters}{'010}
\DeclareMathSymbol{\bsfPsi}{0}{bsfletters}{'011}
\DeclareMathSymbol{\ssfPsi}{0}{ssfletters}{'011}
\DeclareMathSymbol{\bsfOmega}{0}{bsfletters}{'012}
\DeclareMathSymbol{\ssfOmega}{0}{ssfletters}{'012}

\newcommand{\calH}{{\mathcal{H}}}

\newcommand{\calR}{{\mathcal{R}}}

\newcommand{\calU}{{\mathcal{U}}}
\newcommand{\calV}{{\mathcal{V}}}
\newcommand{\calX}{{\mathcal{X}}}
\newcommand{\calY}{{\mathcal{Y}}}

\newcommand{\calZ}{{\mathcal{Z}}}

\begin{document}

\newtheorem{thm}{Theorem} 
\newtheorem{lem}{Lemma}
\newtheorem{prop}{Proposition}
\newtheorem{cor}{Corollary}
\newtheorem{defn}{Definition}
\newtheorem{rem}{Remark}
\newtheorem{ex}{Example}

\newenvironment{example}[1][Example]{\begin{trivlist}
\item[\hskip \labelsep {\bfseries #1}]}{\end{trivlist}}

\renewcommand{\qedsymbol}{ \begin{tiny}$\blacksquare$ \end{tiny} }

\renewcommand{\leq}{\leqslant}
\renewcommand{\geq}{\geqslant}
\renewcommand{\avgH}[1]{{H}\!\left(#1\right)}

\title {Polar Coding for the Broadcast Channel\\ with Confidential Messages:\\ A Random Binning Analogy}

\author{
    \IEEEauthorblockN{R\'{e}mi A. Chou, Matthieu R. Bloch}
    \thanks{R. A. Chou was with the School~of~Electrical~and~Computer~Engineering,~Georgia~Institute~of~Technology, Atlanta,~GA~30332, and is now with the Department of Electrical Engineering, The Pennsylvania State University, Univeristy Park, PA 16802. M. R. Bloch is with the School~of~Electrical~and~Computer~Engineering,~Georgia~Institute~of~Technology, Atlanta,~GA~30332. E-mail : remi.chou@psu.edu; matthieu.bloch@ece.gatech.edu. This work was supported in part by the NSF under Award CCF 1320298 and by ANR with grant 13-BS03-0008. Part of this work has been presented in \cite{Chou15}.}
}
\maketitle

\begin{abstract}
We develop a low-complexity polar coding scheme for the discrete memoryless broadcast channel with confidential messages under strong secrecy and randomness constraints. Our scheme extends previous work by using an optimal rate of uniform randomness in the stochastic encoder, and avoiding assumptions regarding the symmetry or degraded nature of the channels. The price paid for these extensions is that the encoder and decoders are required to share a secret seed of negligible size and to increase the block length through chaining. We also highlight a close conceptual connection between the proposed polar coding scheme and a random binning proof of the secrecy capacity region.
\end{abstract}

\section{Introduction}

With the renewed interest for information-theoretic security, there have been several attempts to develop low-complexity coding schemes achieving the fundamental secrecy limits of the wiretap channel models. In particular, explicit coding schemes based on low-density parity-check codes~\cite{Thangaraj2007,Subramanian2011,Rathi2011}, polar codes~\cite{Mahdavifar11,Sasoglu13,renes2013efficient,Andersson2013}, and invertible extractors~\cite{Hayashi11,Bellare2012} have been successfully developed for special cases of Wyner's model~\cite{Wyner75}, in which the channels are at least required to be symmetric. The recently introduced chaining techniques for polar codes provide, however, a convenient way to construct explicit low-complexity coding schemes for a variety of information-theoretic channel models~\cite{Mondelli14b} without any restrictions on the channels.

In this paper, we develop a low-complexity polar coding scheme for the broadcast channel with confidential messages~\cite{Csiszar78}. We do not make degradation or symmetry assumptions on the communication channel. Moreover, rather than view randomness as a free resource, which could be used to simulate random numbers at arbitrary rate with no cost, we adopt the point of view put forward in~\cite{Watanabe12,Bloch12}, in which any randomness used for stochastic encoding must be explicitly accounted for. In particular, our proposed polar coding scheme exploits the optimal rate of randomness identified in~\cite{Watanabe12} and provides, in addition, a polar coding construction to perform channel prefixing. 

Results related to the present work have been independently and concurrently developed in~\cite{Gulcu14,Wei14}, whose main differences can be summarized as follows. Unlike \cite{Wei14}, our coding scheme does not require that a non-negligible amount of common randomness is shared between the legitimate users as in \cite[Section III-A]{Honda13}, and unlike \cite{Gulcu14}, our coding scheme does not rely on \cite[Theorem 3]{Honda13} and existence, through averaging, of certain deterministic maps. Moreover, in contrast to~\cite{Wei14,Gulcu14}, we consider randomness as a resource and use the optimal amount of local randomness for  the stochastic encoder (see Section~\ref{sec:coding-rates}), we consider auxiliary random variables with non-binary alphabets to achieve the entire region in Theorem~\ref{thm:watanabe} (see Lemma \ref{lemcard_2} and Remark \ref{remcard}), and  we do not assume that channel prefixing can be performed perfectly (see Section~\ref{sec:channel-prefixing}). Note also that \cite{Wei14} only considers weak secrecy. Consequently, our coding scheme and proofs are different from~\cite{Wei14,Gulcu14}. Remark also that, in our encoding scheme, we do not use maximum a posteriori (MAP) decisions\footnote{We refer the reader to \cite{Chou15f} for additional details on MAP decisions in encoding and decoding of polar codes.} in the same way as in \cite{Wei14,Gulcu14}. 
When specialized to Wyner's wiretap model, our scheme is also related to~\cite{renes2013efficient}, but with a number of notable distinctions. Specifically, while no pre-shared secret seed is required in~\cite{renes2013efficient}, the coding scheme therein relies on a two-layer construction for which no efficient code construction is presently known~\cite[Section 3.3]{renes2013efficient}. In contrast, our coding scheme requires a pre-shared secret seed, but at the benefit of only using a single layer of polarization. 

We summarize a comparison between our result specialized to the wiretap channel model and \cite{renes2013efficient,Gulcu14,Wei14} in Figure~\ref{fig_dif}.
\begin{figure}
\small
\begin{center}
    \begin{tabular}{ c | c | c | c | c|} 
    \cline{2-5}
      & \cite{renes2013efficient} & \cite{Gulcu14} & \cite{Wei14} & \phantom{$\hat{l}$}This paper\phantom{$\hat{l}$} \\ \hline
   \multicolumn{1}{|c|}{\phantom{$\hat{l}$}1)} & \cmark & \xmark & \cmark  & \cmark \\ \hline
       \multicolumn{1}{|c|}{\phantom{$\hat{l}$}2)}  & \cmark & \cmark & \xmark  & \cmark \\ \hline
     \multicolumn{1}{|c|}{ \phantom{$\hat{l}$}3)} & \cmark & \cmark & \xmark  & \cmark \\ \hline
                 \multicolumn{1}{|c|}{ \phantom{$\hat{l}$}4)}
& \xmark & \xmark & \xmark  & \cmark \\ \hline    
 \multicolumn{1}{|c|}{ \phantom{$\hat{l}$}5)} & \xmark & \xmark & \xmark  & \cmark \\ \hline    
    \multicolumn{1}{|c|}{ \phantom{$\hat{l}$}6)} & \xmark & \xmark & \xmark  & \cmark \\ \hline    
        \multicolumn{1}{|c|}{ \phantom{$\hat{l}$}7)} & \xmark & \cmark & \cmark  & \cmark \\ \hline    
            \end{tabular}
\end{center}
\normalsize
\caption{Summary of differences between the present work and related polar coding schemes for arbitrary discrete memoryless wiretap channels~\cite{renes2013efficient,Gulcu14,Wei14}. 1) holds when the coding scheme is explicit and does not rely on existence, through averaging, of certain deterministic maps as in \cite[Theorem~3]{Honda13}, 2) holds when the coding scheme does not rely on a non-negligible amount of common randomness shared between the legitimate users as in \cite[Section III.A]{Honda13}, 3) holds when strong secrecy is considered, 4) holds when non-binary auxiliary random variables are considered -- see Lemma \ref{lemcard_2}, 5) holds when the optimal amount of local randomness is used at the encoder, 6) holds when it is not assumed that channel prefixing can be perfectly performed, 7) holds when an efficient code construction is known.} \label{fig_dif}
\end{figure}

We summarize our contributions as follows.
\begin{itemize}
\item For the broadcast channel with confidential messages, we propose an explicit low-complexity and capacity achieving coding scheme under strong secrecy. Moreover, we do not make symmetry or degradation assumptions on the communication channel. Our result particularizes to the wiretap channel model to also provide an explicit low-complexity and capacity achieving coding scheme under strong secrecy.\footnote{Although no secrecy constraint holds on the common messages for a broadcast channel model, the latter introduces additional difficulties  in the security analysis, compared to a point-to-point wiretap channel model, because of our chaining constructions; see Figure \ref{figFGD2}.}
	\item To the best of our knowledge, the parallel between random binning and polar codes made in the manuscript does not explicitly appear elsewhere. This conceptual consideration also has direct implications for the study of our coding scheme. Specifically, it stresses the fact that the distribution induced by the encoder must be precisely analyzed to rigorously assess reliability and secrecy.
	\item We develop a scheme that uses the minimal amount of local randomness required in the stochastic encoding. %
	\item We consider polar coding for channel prefixing and do not assume that this operation can be perfectly~realized. %
\end{itemize}

The remaining of the paper is organized as follows. Section~\ref{sec:broadc-chann-with} formally introduces the notation and the model under investigation. Section~\ref{Sec_CS} develops a random binning proof of the results in~\cite{Watanabe12}, which serves as a guideline for the design of the polar coding scheme. Section~\ref{sec:polar-coding-schem} describes the proposed polar coding scheme, while Section~\ref{sec:analys-polar-coding} provides its detailed analysis. Section~\ref{sec:conclusion} offers some concluding remarks.

\section{Broadcast channel with confidential messages and constrained randomization}
\label{sec:broadc-chann-with}

\subsection{Notation} 
We define the integer interval $\llbracket a,b \rrbracket$, as the set of integers between $\lfloor a \rfloor$ and $\lceil b \rceil$.  For $n \in \mathbb{N}$ and $N \triangleq 2^n$, we let $G_n \triangleq  \left[ \begin{smallmatrix}
       1 & 0            \\[0.3em]
       1 & 1 
     \end{smallmatrix} \right]^{\otimes n} $ be the source polarization transform defined in~\cite{Arikan10}. Let the components of a vector, $X^{1:N}$, of size $N$, be denoted by superscripts, i.e., $X^{1:N} \triangleq (X^1 , X^2, \ldots, X^{N})$. For any set of indices $\mathcal{I} \subseteq \llbracket 1,N \rrbracket$, we define $X^{1:N}[\mathcal{I}] \triangleq \{ X^i \}_{i\in \mathcal{I}}$. We also use the notation $\mathcal{S}^c$ to the denote the complement in~$\llbracket 1,N \rrbracket$ of any subset $\mathcal{S}$ of $\llbracket 1,N \rrbracket$.
Unless specified otherwise, capital letters designate random variables, whereas lowercase letters designate realizations of associated random variables, e.g., $x$ is a realization of the random variable $X$. When the context makes clear that we are dealing with vectors, we write $X^N$ in place of $X^{1:N}$. Let $\mathbb{V}(\cdot, \cdot)$ and $\mathbb{D}(\cdot || \cdot)$ denote the variational distance and the divergence, respectively, between two distributions. Finally, we define the indicator function  $\mathds{1}\{ \omega \}$, which is equal to $1$ if the predicate $\omega$ is true and $0$ otherwise.

\subsection{Channel model and capacity region}
\label{Sec_PS}
We consider the problem of secure communication over a discrete memoryless broadcast channel $(\mathcal{X}, p_{YZ|X}, \mathcal{Y},\mathcal{Z})$ illustrated in Figure~\ref{figBCC}. %
The marginal probabilities $p_{Y|X}$ and $p_{Z|X}$ define two \acp{DMC} $(\mathcal{X}, p_{Y|X}, \mathcal{Y})$ and $(\mathcal{X}, p_{Z|X}, \mathcal{Z})$, which we refer to as Bob's channel and Eve's channel, respectively. 

\begin{figure}
\centering
  \includegraphics[width=8.7cm]{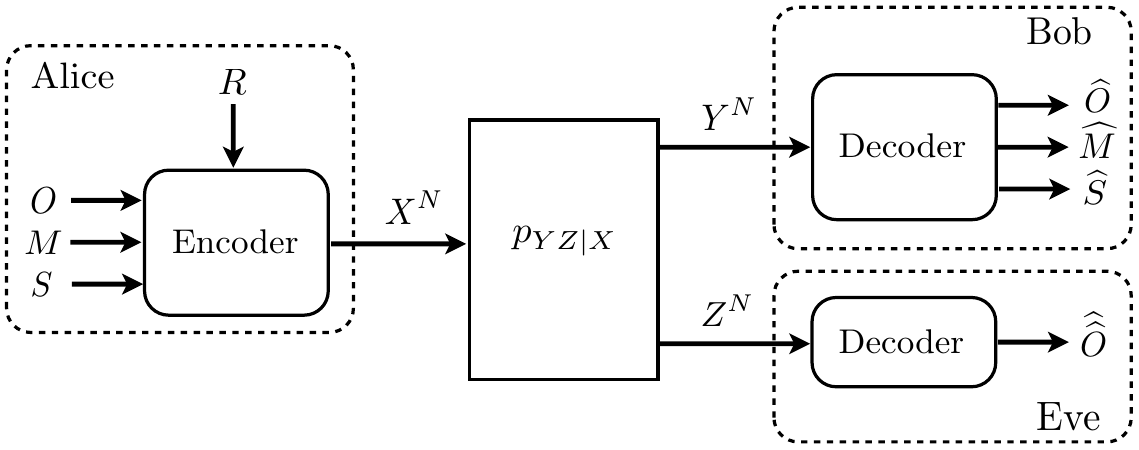}
  \caption{Communication over a broadcast channel with confidential messages. $O$ is a common message that must be reconstructed by both Bob and Eve. $S$ is a confidential message that must be reconstructed by Bob and kept secret from Eve. $M$ is a private message that Alice wishes to send to Bob without secrecy constraint, i.e., $M$ is not required to be reconstructed by Eve and is not required to be kept secret from Eve. $R$ represents an additional randomization sequence used at the encoder. }
  \label{figBCC}
\end{figure}

\begin{defn}
A $(2^{NR_O},2^{NR_M},2^{NR_S},2^{NR_R},N)$ code $\mathcal{C}_N$ for the broadcast channel consists of
\begin{itemize}
\item a common message set $\mathcal{O} \eqdef \llbracket 1 , 2^{NR_O} \rrbracket$;
\item a private message set $\mathcal{M} \eqdef \llbracket 1 , 2^{NR_M} \rrbracket$;
\item a confidential message set $\mathcal{S} \eqdef \llbracket 1 , 2^{NR_S} \rrbracket$;
\item a randomization sequence set $\mathcal{R} \eqdef \llbracket 1 , 2^{NR_R} \rrbracket$;
\item an encoding function $f:  \mathcal{O} \times \mathcal{M} \times \mathcal{S}
\times \mathcal{R}  \to \mathcal{X}^N$, which maps the messages $(o,m,s)$ and the randomness $r$ to a codeword $x^N$;
\item a decoding function $g:  \mathcal{Y}^N \to \mathcal{O} \times \mathcal{M} \times \mathcal{S}$, which maps each observation of Bob's channel $y^N$ to the messages $(\hat{o},\hat{m},\hat{s})$;
\item a decoding function $h:  \mathcal{Z}^N \to \mathcal{O} $, which maps each observation of Eve's channel $z^N$ to the message~$\hat{\hat{o}}$.
\end{itemize}
\end{defn}

\begin{rem}
The randomization sequence required at the encoder is used for prefixing and is not needed at the decoder. We refer to it as ``local randomness."%
\end{rem}

For uniformly distributed $O$, $M$, $S$, and $R$, the performance of a $(2^{NR_O},2^{NR_M},2^{NR_S},2^{NR_R},N)$ code $\mathcal{C}_N$ for the broadcast channel is measured in terms of its probability of error
$$
\mathbf{P}_e(\mathcal{C}_N) \triangleq \mathbb{P}\left[ \left\{ (\widehat{O},\widehat{M},\widehat{S}) \neq (O,M,S) \right\} \cup \left\{ \widehat{\widehat{O}} \neq O \right\}\right],
$$
and its leakage of information about the confidential message to Eve
$$
\mathbf{L}_e(\mathcal{C}_N) \triangleq I(S;Z^N).
$$

\begin{defn}
A rate tuple $(R_O,R_M,R_S,R_R)$ is achievable for the broadcast channel if there exists a sequence of $(2^{NR_O},2^{NR_M},2^{NR_S},2^{NR_R},N)$ codes $\{ \mathcal{C}_N \}_{N \geq 1}$ such that
\begin{align*}
\lim_{N \to \infty} \mathbf{P}_e(\mathcal{C}_N)& =0 \text{ (reliability condition}),\\
\lim_{N \to \infty}  \mathbf{L}_e(\mathcal{C}_N)& =0 \text{ (strong secrecy}).
\end{align*}
The achievable region $\mathcal{R}_{\textup{BCC}}$ is defined as the closure of the set of all achievable rate quadruples. 
\end{defn}

\begin{rem}
We require strong secrecy, as opposed to weak secrecy which would require  
$$
\lim_{N \to \infty}  \frac{\mathbf{L}_e(\mathcal{C}_N)}{N}=0.
$$
 Weak secrecy can often be analyzed through an astute use of Fano's inequality \cite{Fano61}. Strong secrecy usually requires more involved proof techniques but is perhaps a more meaningful secrecy metric as discussed in \cite{Maurer00}.
\end{rem}

The exact characterization of $\mathcal{R}_{\textup{BCC}}$ was obtained in~\cite{Watanabe12}.
\begin{thm}[{\cite{Watanabe12}}]
\label{thm:watanabe}
$\mathcal{R}_{\textup{BCC}}$ is the closed convex set consisting of the quadruples $(R_O,R_M,R_S,R_R)$ for which there exist auxiliary random variables $(U,V)$ such that $U - V - X - (Y,Z)$, $|\mathcal{U}| \leq |\mathcal{X}|+3$, $|\mathcal{V}| \leq (|\mathcal{X}|+3)(|\mathcal{X}|+1)$, and
\begin{align*}
R_O &\leq \min[ I(U;Y), I(U;Z) ],  \\
R_O + R_M + R_S &\leq I(V;Y|U) + \min[ I(U;Y), I(U;Z)],  \\
R_S &\leq I(V;Y|U) - I(V;Z|U),  \\
R_M + R_R &\geq I(X;Z|U),  \\
R_R &\geq I(X; Z|V). 
\end{align*}
\end{thm}
The main contribution of the present work is to develop a polar coding scheme achieving the rates in $\mathcal{R}_{\textup{BCC}}$.

\section{A binning approach to code design: from random binning to polar binning} 
\label{Sec_CS}

In this section, we argue that our construction of polar codes for the broadcast channel with confidential messages is essentially the constructive counterpart of a \emph{random binning} proof of the region $\calR_{\textup{BCC}}$. While \emph{random coding} is often the natural tool to address channel coding problems, random binning is already found in~\cite{Csiszar1996} to establish the strong secrecy of the wiretap channel, and is the tool of choice in quantum information theory~\cite{Renes2011}; there has also been a renewed interest for random binning proofs in multi-user information theory, motivated in part by~\cite{yassaee2014achievability}. In Section~\ref{sec:rand-binn-secure}, we sketch a random binning proof of the characterization of $\calR_{\textup{BCC}}$ established in~\cite{Watanabe12}, which may be viewed as a refinement of the analysis in~\cite{yassaee2014achievability} to obtain a more precise characterization of the stochastic encoder. Section~\ref{sec:rand-binn-secure} does not involve polar codes and does not contain new results, but we use this alternative proof in Section~\ref{sec:binning-with-polar} to obtain high-level insight into the construction of polar codes. The main benefit is to clearly highlight the crucial steps of the construction in Section~\ref{sec:polar-coding-schem} and of its analysis in Section~\ref{sec:analys-polar-coding}. In particular, the rate conditions developed in the random binning proof of Section~\ref{sec:rand-binn-secure} directly translate into the definition of the polarization sets in Section~\ref{sec:binning-with-polar}.

\subsection{Information-theoretic random binning}
\label{sec:rand-binn-secure}
Information-theoretic random binning proofs rely on the following well-known lemmas -- see, for instance, \cite{Csiszar1996,Renes2011,yassaee2014achievability} for a proof. We use the notation $\delta(N)$ to denote an unspecified positive function of $N$ that vanishes as $N$ goes to infinity.

\begin{lem}[Source-coding with side information]
\label{lm:1}
  Consider a \ac{DMS} $(\calX\times\calY,p_{XY})$. For each $x^N\in\calX^N$, assign an index $\Phi(x^N)\in\intseq{1}{2^{NR}}$ uniformly at random. If $R>\avgH{X|Y}$, then $\exists N_0$ such that $\forall N\geq N_0$, there exists a deterministic function $$g_N:\intseq{1}{2^{NR}}\times\calY^N\rightarrow \calX^N:(\Phi(x^N),y^N)\mapsto \hat{x}^N$$ such that
  \begin{align*}
\mathbb{E}_{\Phi} \left[ \mathbb{P} \left[ X^N \neq g_N(\Phi(X^N),Y^N)\right] \right]\leq \delta(N).
  \end{align*}
\end{lem}

\begin{lem}[Privacy amplification, channel intrinsic randomness, output statistics of random binning]
\label{lm:2}
  Consider a \ac{DMS} $(\calX\times\calZ,p_{XZ})$ and let $\epsilon>0$. For each $x^N\in\calX^N$, assign an index $\Psi(x^N)\in\intseq{1}{2^{NR}}$ uniformly at random. Denote by $q_U$ the uniform distribution on $\intseq{1}{2^{NR}}$. 
  
  If $    R<\avgH{X|Z}$, then $\exists N_0$ such that $\forall N\geq N_0$
  \begin{align*}
    \mathbb{E}_{\Psi} \left[\V{p_{\Psi(X^N)Z^N},q_U p_{Z^N}} \right]\leq\delta(N).
  \end{align*}
\end{lem}
One may obtain more explicit results regarding the convergence to zero in Lemma~\ref{lm:1} and Lemma~\ref{lm:2}, but we ignore this for brevity.

The principle of a random binning proof of Theorem~\ref{thm:watanabe} is to consider a \ac{DMS} $(\calU\times\calV\times\calX\times\calY\times\calZ,p_{UVXYZ})$ such that $U-V-X-YZ$, and to assign two types of indices to source sequences by random binning. The first type identifies subsets of sequences that play the roles of codebooks, while the second type labels sequences with indices that can be thought of as messages. As explained in the next paragraphs, the crux of the proof is to show that the binning can be ``inverted,'' so that the sources may be generated from independent choices of uniform codebooks and messages.\smallskip

\noindent\textbf{Common message encoding.} We introduce two indices $\psi^U\in\intseq{1}{2^{N\rho_U}}$ and $o\in\intseq{1}{2^{NR_O}}$ by random binning on $u^N$ such that:
\begin{itemize}
\item $\rho_U>\max\left(\avgH{U|Y},\avgH{U|Z}\right)$, so that Lemma~\ref{lm:1} ensures\footnote{Apply the substitutions $R \leftarrow \rho_U$, $\Phi(X^N) \leftarrow \Psi^U$, $X \leftarrow U$, and $Y \leftarrow (Y \text{ or } Z)$.} that the knowledge of $\Psi^U$ allows Bob and Eve to reconstruct $U^N$ with high probability knowing $Y^N$ or $Z^N$, respectively;
\item $\rho_U+R_O<\avgH{U}$, so that Lemma~\ref{lm:2} ensures\footnote{Apply the substitutions $R \leftarrow (\rho_U + R_O)$, $\Psi(X^N) \leftarrow (\Psi^U,O)$, $X \leftarrow U$, and $Z \leftarrow \emptyset$.} that $\Psi^U$ and $O$ are almost uniformly distributed and independent of each other.
\end{itemize}
The binning scheme induces a joint distribution $p_{U^N\Psi^UO}$. To convert the binning scheme into a channel coding scheme, Alice operates as follows. Upon sampling indices $\widetilde{\psi}^U\in\intseq{1}{2^{N\rho_U}}$ and $\widetilde{o}\in\intseq{1}{2^{NR_O}}$ from independent uniform distributions, Alice \emph{stochastically} encodes them into a sequence $\widetilde{u}^{N}$ drawn according to $p_{U^N|\Psi^UO}(\widetilde{u}^N|\widetilde{\psi}^U,\widetilde{o})$. The choice of rates above guarantees that the joint distribution $p_{\widetilde{U}^N\widetilde{\Psi}^U\widetilde{O}}$ approximates the distribution $p_{U^N\Psi^UO}$ in variational distance, so that disclosing $\widetilde{\psi}^U$ allows Bob and Eve to decode the sequence $\widetilde{u}^N$. \smallskip

\noindent\textbf{Secret and private message encoding.} Following the same approach, we introduce three indices $\psi^{V|U}\in\intseq{1}{2^{N\rho_{V|U}}}$, $s\in\intseq{1}{2^{NR_S}}$, and $m\in\intseq{1}{2^{NR_M}}$ by random binning on $v^N$ such that
\begin{itemize}
\item $\rho_{V|U}>\avgH{V|UY}$, to ensure\footnote{By Lemma~\ref{lm:1} with the substitutions $R \leftarrow \rho_{V|U}$, $\Phi(X^N) \leftarrow \Psi^{V|U}$,  $X \leftarrow V$, and $Y \leftarrow (U,Y)$.} that knowing $\Psi^{V|U}$, $U^N$, and $Y^N$, Bob may reconstruct $V^N$;
\item $\rho_{V|U}+R_S<\avgH{V|UZ}$ and $\rho_{V|U}+R_S+R_M<\avgH{V|U}$ to ensure\footnote{By Lemma~\ref{lm:2} with the substitutions $R \leftarrow (\rho_{V|U} + R_S)$, $\Psi(X^N) \leftarrow (\Psi^{V|U},S)$, $X \leftarrow V$, and $Z \leftarrow (U,Z)$, and with the substitutions $R \leftarrow (\rho_{V|U} + R_S + R_M)$, $\Psi(X^N) \leftarrow (\Psi^{V|U},S,M)$, $X \leftarrow V$, and $Z \leftarrow U$.} that the indices are almost uniformly distributed and independent of each other, as well as of $U^N$ or $(U^N,Z^N)$ for the secret message $S$.
\end{itemize}
The binning scheme induces a joint distribution $p_{V^NU^N\Psi^{V|U}SM}$. To obtain a channel coding scheme, Alice encodes the realizations of independent and uniformly distributed indices $\widetilde{\psi}^{V|U}\in\intseq{1}{2^{N\rho_{V|U}}}$, $\widetilde{s}\in\intseq{1}{2^{NR_S}}$,  $\widetilde{m}\in\intseq{1}{2^{NR_M}}$, and the sequence $ \widetilde{u}^N$, into a sequence $\widetilde{v}^N$ drawn according to the distribution $p_{V^N|U^N\Psi^{V|U}SM}(\widetilde{v}^N|\widetilde{u}^N,\widetilde{\psi}^{V|U},\widetilde{s},\widetilde{m})$. The resulting joint distribution is again a close approximation of $p_{V^NU^N\Psi^{V|U}SM}$, so that the scheme inherits the reliability and secrecy properties of the random binning scheme upon disclosing $\widetilde{\psi}^{V|U}$. \smallskip

\noindent\textbf{Channel prefixing.} Finally, we introduce the indices $\psi^{X|V}\in\intseq{1}{2^{N\rho_{X|V}}}$ and $r\in\intseq{1}{2^{NR_R}}$ by random binning on $x^N$ such that
\begin{itemize}
\item $\rho_{X|V}<\avgH{X|VZ}$ to ensure\footnote{By Lemma~\ref{lm:2} with the substitutions $R \leftarrow \rho_{X|V}$, $\Psi(X^N) \leftarrow \Psi^{X|V}$, and $Z \leftarrow (V,Z)$.} that $\Psi^{X|V}$ is independent of $V^N$ and $Z^N$;
\item $\rho_{X|V}+R_R<\avgH{X|V}$ to ensure\footnote{By Lemma~\ref{lm:2} with the substitutions $R \leftarrow (\rho_{X|V} + R_R)$, $\Psi(X^N) \leftarrow (\Psi^{X|V},R)$, and $Z \leftarrow V$.} that the indices are almost uniformly distributed and independent of each other, as well as of $V^N$.
\end{itemize}
The binning scheme induces a joint distribution $p_{X^NV^NU^N\Psi^{X|V}R}$. To obtain a channel prefixing scheme, Alice encodes the realizations of uniformly distributed indices $\widetilde{\psi}^{X|V}$ and $\widetilde{r}$, and the previously obtained $\widetilde{v}^N$ into a sequence $\widetilde{x}^N$ drawn according to $p_{X^N|V^N\Psi^{X|V}R}(\widetilde{x}^N|\widetilde{v}^N\widetilde{\psi}^{X|V}\widetilde{r})$. The resulting joint distribution induced is once again a close approximation of $p_{X^NV^NU^N\Psi^{X|V}R}$.\smallskip

\noindent\textbf{Chaining to de-randomize the codebooks.} The downside of the schemes described earlier is that they require sharing the indices $\widetilde{\psi}^{U}$, $\widetilde{\psi}^{V|U}$, and $\widetilde{\psi}^{X|V}$, identifying the codebooks between Alice, Bob, and Eve; however, the rate cost may be amortized by reusing the \emph{same} indices over sequences of $k$ blocks. Specifically, the union bound shows that the average error probability over $k$ blocks is at most $k$ times that of an individual block, and a hybrid argument shows that the information leakage over $k$ blocks is at most $k$ times that of an individual block. Consequently, for $k$ and $N$ large enough, the impact on the transmission rates is negligible.\smallskip

\noindent\textbf{Total amount of randomness.} The total amount of randomness required for encoding includes not only the explicit random numbers used for channel prefixing but also all the randomness required in the stochastic encoding to approximate the source distribution. One can show that the rate randomness specifically used in the stochastic encoding is negligible; we omit the proof of this result for random binning, but this is analyzed precisely for polar codes in Section~\ref{sec:analys-polar-coding}. 

By combining all the rate constraints above and performing Fourier-Motzkin elimination, one recovers the rates in Theorem~\ref{thm:watanabe}.

\subsection{Binning with polar codes}
\label{sec:binning-with-polar} 

The main observation to translate the analysis of Section~\ref{sec:rand-binn-secure} into a polar coding scheme is that Lemma~\ref{lm:1} and Lemma~\ref{lm:2} have the following counterparts in terms of source polarization. 
\begin{lem}[adapted from~{\cite{Arikan10}}]
  \label{lm:3}
  Consider a \ac{DMS} $(\calX\times\calY,p_{XY})$. For each $x^{1:N}\in\mathbb{F}_2^N$ polarized as $u^{1:N}\triangleq x^{1:N} G_n$, let $u^{1:N}[\calH_{X|Y}]$ denote the high entropy bits of $u^{1:N}$ in positions $\calH_{X|Y}\eqdef\{i\in\llbracket 1,N \rrbracket:\avgH{U^i|U^{1:i-1}Y^{1:N}}>\delta_N\}$ and $\delta_N\eqdef 2^{-N^\beta}$ with $\beta\in]0,\tfrac{1}{2}[$. For every $i\in\llbracket 1,N \rrbracket$, sample $\widetilde{u}^{1:N}$ from the distribution
  \begin{multline*}
    \widetilde{p}_{U^i|U^{1:i-1}}(\widetilde{u}^i|\widetilde{u}^{1:i-1}) \\
    \eqdef
      \begin{cases}
        \mathds{1}\left\{\widetilde{u}^i=u^i\right\}&\text{if }i\in\calH_{Y|X}\\
          p_{U^i|U^{1:i-1}Y^{1:N}}(\widetilde{u}^i|\widetilde{u}^{1:i-1}y^{1:N}) &\text{if }i\in\calH_{Y|X}^c
        \end{cases},
    \end{multline*}
    and create $\widetilde{x}^{1:N}=\widetilde{u}^{1:N}G_n$. Then,
    \begin{align*}
      \mathbb{P} \left[ \widetilde{X}^{1:N} \neq X^{1:N} \right] = O(N\delta_N),
    \end{align*}
    and $\displaystyle\lim_{N\rightarrow\infty} \frac{1}{N}\card{\calH_{X|Y}}=\avgH{X|Y}$.
  \end{lem}
In other words, the high entropy bits in positions $\calH_{X|Y}$ play the same role as the random binning index in Lemma~\ref{lm:1}. However, note that the construction of $\widetilde{x}^{1:N}$ in Lemma~\ref{lm:3} is explicitly stochastic. 

\begin{lem}[adapted from~{\cite{Chou14rev}}]
\label{lm:4}
    Consider a \ac{DMS} $(\calX\times\calZ,p_{XZ})$. For each $x^{1:N}\in\mathbb{F}_2^N$ polarized as $u^{1:N}\triangleq x^{1:N} G_n$, let $u^{1:N}[\calV_{X|Z}]$ denote the very high entropy bits of $u^{1:N}$ in positions $\calV_{X|Z}\eqdef\{i\in\llbracket 1,N \rrbracket:\avgH{U^i|U^{1:i-1}Z^{1:N}}>1-\delta_N\}$ and $\delta_N\eqdef 2^{-N^\beta}$ with $\beta\in]0,\tfrac{1}{2}[$. Denote by $q_{U}$ the uniform distribution over $ \llbracket 1 ,2^{|\calV_{X|Z}|} \rrbracket$. Then,
\begin{align*}
  \V{p_{U^{1:N}[\calV_{X|Z}]Z^{1:N}},q_{U}p_{Z^{1:N}}} = O(\sqrt{N\delta_N}),
\end{align*}
and $\displaystyle\lim_{N\rightarrow\infty}\frac{1}{N}\card{\calV_{X|Z}}=\avgH{X|Z}$ by \cite[Lemma 1]{Chou14rev}.
\end{lem}
The very high entropy bits in positions $\calV_{X|Z}$ therefore play the same role as the random binning index in Lemma~\ref{lm:2}.

Intuitively, information theoretic constraints resulting from Lemma~\ref{lm:1} translate into the use of ``high entropy'' sets $\calH$, while those resulting from Lemma~\ref{lm:2}  translate into the use of ``very high entropy'' sets $\calV$. However, unlike the indices resulting from random binning, the high entropy and very high entropy sets may not necessarily be aligned, and the precise design of a polar coding scheme requires more care.

In the remainder of the paper, we consider a \ac{DMS} $(\calU\times\calV\times\calX\times\calY\times\calZ,p_{UVXYZ})$ such that $U-V-X-YZ,\text{ and }  I(V;Y|U) - I(V;Z|U) >0,\footnote{This avoids the trivial case of $R_S=0$ in Theorem \ref{thm:watanabe}, i.e., no secret information can be transmitted over the channel.}$ $|\mathcal{X}| = q^{(X)}$, with $q^{(X)}$ a prime number, $|\mathcal{U}|=q^{(U)}$, with $q^{(U)}$ the smallest prime number larger than $q^{(X)}+3$, and $ |\mathcal{V}| = q^{(V)}$, with $q^{(V)}$ the smallest prime number larger than $(q^{(X)}+3)(q^{(X)}+1)$. We also assume without loss of generality $I(U;Y) \leq I(U;Z)$, since the case $I(U;Y) > I(U;Z)$ is obtained by exchanging the role of $Y$ and $Z$ in the encoding scheme for the common messages, and by exchanging the role of Bob and Eve in the decoding of the common messages.

\noindent\textbf{Common message encoding.} Define the polar transform of $U^{1:N}$, as $A^{1:N} \triangleq U^{1:N} G_n$ and the associated sets
\begin{align}
 {\mathcal{H}}_U & \triangleq \left\{ i \in \llbracket 1,N \rrbracket  : H( A^i | A^{1:i-1}) >  \delta_N  \right\}, \label{eq:common_msg_sets_1}\\
{\mathcal{V}}_U & \triangleq \left\{ i \in \llbracket 1,N \rrbracket  : H( A^i | A^{1:i-1}) > \log_2(q^{(U)})-  \delta_N  \right\},\\
 {\mathcal{H}}_{U|Y} & \triangleq \left\{ i \in \llbracket 1,N \rrbracket  :  H( A^i | A^{1:i-1} Y^{1:N}) > \delta_N \right\},\\
 \mathcal{H}_{U|Z} & \triangleq \left\{ i \in \llbracket 1,N \rrbracket  :  H( A^i | A^{1:i-1} Z^{1:N}) > \delta_N \right\}. \label{eq:common_msg_sets_6}
\end{align}
If we could guarantee\footnote{In general, one only has ${\mathcal{V}}_U \subseteq {\mathcal{H}}_U$, ${\mathcal{H}}_{U|Y} \subseteq {\mathcal{H}}_U$, and ${\mathcal{H}}_{U|Z} \subseteq {\mathcal{H}}_U$.} that $\calH_{U|Z}\subseteq \calH_{U|Y}\subseteq\calV_{U}$, then we could directly mimic the information-theoretic random binning proof. We would use random $q^{(U)}$-ary symbols in positions $\calH_{U|Z}$ to identify the code, random $q^{(U)}$-ary symbols in positions $\calV_U\setminus \mathcal{H}_{U|Z}$ for the message, successive cancellation encoding to compute the $q^{(U)}$-ary symbols in positions $\calV_U^c$ and approximate the source distribution, and chaining to amortize the rate cost of the $q^{(U)}$-ary symbols in positions $\calH_{U|Z}$. Unfortunately, the inclusion $\calH_{U|Y}\subseteq \calH_{U|Z}$ is not true in general, and one must also use chaining as to ``realign'' the sets of indices. Furthermore, only the inclusions $\calH_{U|Z}\subseteq \calH_U$ and $\calH_{U|Y}\subseteq \calH_U$ are true in general, so that the $q^{(U)}$-ary symbols in positions $\calH_{U|Z}\cap\calV_U^c$ and $\calH_{U|Y}\cap\calV_U^c$ must be transmitted separately. The precise coding scheme is detailed in Section~\ref{sec:comm-mess-encod}.

\noindent\textbf{Secret and private messages encoding.} Define the polar transform of $V^{1:N}$ as $B^{1:N} \triangleq V^{1:N} G_n$ and the associated sets
\begin{align}
\mathcal{H}_{V|UY}  & \triangleq   \left\{ i \in \llbracket 1,N \rrbracket:  H( B^i | B^{1:i-1} U^{1:N}Y^{1:N}) > \delta_N  \right\}, \label{eq:private_msg_sets_1} \\ 
\mathcal{V}_{V|U}  &\triangleq  \left\{ i \in \llbracket 1,N \rrbracket:  H( B^i | B^{1:i-1} U^{1:N})\right.\nonumber\\
 &\left. \phantom{mmmmlmmmmmmm}>  \log_2(q^{(V)}) - \delta_N  \right\},  \\ \nonumber
\mathcal{V}_{V|UZ} & \triangleq  \left\{ i \in \llbracket 1,N \rrbracket: H( B^i | B^{1:i-1} U^{1:N}Z^{1:N}) \right.\nonumber\\
 &\left. \phantom{mmmmlmmmmmmm} > \log_2(q^{(V)}) - \delta_N  \right\}, \displaybreak[0]\\
 \mathcal{V}_{V|UY} &   \triangleq  \left\{ i \in \llbracket 1,N \rrbracket: H( B^i | B^{1:i-1} U^{1:N}Y^{1:N}) \right.\nonumber\\
 &\left. \phantom{mmmmmlmmmmmm}> \log_2(q^{(V)}) - \delta_N  \right\}.  \label{eq:private_msg_sets_5}
\end{align}
If the inclusion $\calH_{V|UY}\subseteq \calV_{V|UZ}$ were true,\footnote{In general, we only have ${\mathcal{V}}_{V|UZ} \subseteq {\mathcal{V}}_{V|U}$, ${\mathcal{V}}_{V|UY} \subseteq {\mathcal{H}}_{V|UY}$, and ${\mathcal{V}}_{V|UY} \subseteq {\mathcal{V}}_{V|U}$.} then we would place random $q^{(V)}$-ary symbols identifying the codebook in positions $\calH_{V|UY}$, random $q^{(V)}$-ary symbols describing the secret message  in positions $\calV_{V|UZ}\setminus \calH_{V|UY}$, random $q^{(V)}$-ary symbols describing the private message in positions $\calV_{V|U}\setminus \calV_{V|UZ}$, use successive cancellation encoding to compute the $q^{(V)}$-ary symbols in positions $\calV_{V|U}^c$ and approximate the source distribution, and use chaining to amortize the rate cost of the $q^{(V)}$-ary symbols in positions $\calH_{V|UY}$. This is unfortunately again not directly possible in general, and one needs to exploit chaining to realign the indices, and transmit the $q^{(V)}$-ary symbols in positions $\calH_{V|UY}\cap\calV_{V|U}^c$ separately and secretly to Bob. The precise coding scheme is detailed in Section~\ref{sec:secr-priv-mess}.

\noindent\textbf{Channel prefixing.} Finally, define the polar transform of $X^{1:N}$ as $T^{1:N} \triangleq X^{1:N} G_n$ and the associated sets
\begin{align}
 \mathcal{V}_{X|V} & \triangleq \left\{ i \in \llbracket 1,N \rrbracket: H( T^i | T^{1:i-1} V^{1:N}) \right.\nonumber\\
& \left. \phantom{mmmmmmmmmmml}> \log_2(q^{(X)}) - \delta_N  \right\}, \label{eq:randomization_msg_sets_1}\\
 \mathcal{V}_{X|VZ} & \triangleq \left\{ i \in \llbracket 1,N \rrbracket:  H( T^i | T^{1:i-1} V^{1:N}Z^{1:N}) \right.\nonumber\\
& \left. \phantom{mmmmmlmmmmmm} > \log_2(q^{(X)}) - \delta_N  \right\}. \label{eq:randomization_msg_sets_2}
\end{align}
Note that $\mathcal{V}_{X|V} \subseteq \mathcal{V}_{X|VZ}$. One performs channel prefixing by placing random $q^{(X)}$-ary symbols identifying the code in positions $\calV_{X|VZ}$, random $q^{(X)}$-ary symbols describing the randomization sequence in positions $\calV_{X|V}\setminus \calV_{X|VZ}$, and using successive cancellation encoding to compute the $q^{(X)}$-ary symbols in positions $\calV_{X|V}^c$ and approximate the source distribution. Chaining is finally used to amortize the cost of randomness for describing the code.  The precise coding scheme is detailed in Section~\ref{sec:channel-prefixing}.
\begin{rem}
Although we only formally prove it for the model considered in this paper, we conjecture that any results obtained from random binning could be derived using source polarization as a constructive and low-complexity alternative. This conjecture has been shown to hold for secret-key generation~\cite{Chou14rev}, uniform compression \cite[Section IV-B]{chou2015coding}, strong coordination \cite{Chou15b}, and channel resolvability \cite{Chou15b}. 
\end{rem}

\section{Polar coding scheme}
\label{sec:polar-coding-schem}

In this section, we describe the details of the polar coding scheme resulting from the discussion of the previous section. Recall that the joint probability distribution $p_{UVXYZ}$ of the original source is fixed and defined as in Section \ref{sec:binning-with-polar}. As alluded to earlier, we perform the encoding over $k$ blocks of size $N$. We use the subscript $i \in \llbracket 1, k \rrbracket$ to denote random variables associated to encoding Block $i$. The chaining constructions corresponding to the encoding of the common, secret, and private messages, and randomization sequence, are described in Section~\ref{sec:comm-mess-encod}, Section~\ref{sec:secr-priv-mess}, and Section~\ref{sec:channel-prefixing}, respectively. Although each chaining is described independently, all messages should be encoded in every block before moving to the next. Specifically, in every block $i\in\intseq{1}{k-1}$, Alice successively encodes the common message, the secret and private messages, and performs channel prefixing, before she moves to the next block $i+1$.

\begin{rem}
In the following, we construct random variables whose distributions approach target distributions. We use the tilde in the notation for these random variables to display this intention. For instance, we construct the random variable $\widetilde{U}^{1:N}$ with distribution $\widetilde{p}_{U^{1:N}}$ such that $\widetilde{p}_{U^{1:N}}$ approaches the distribution $p_{U^{1:N}}$ of the random variable $U^{1:N}$. We provide a precise analysis of the variational distance between the distribution of the ``tilded" random variables and the targeted distributions in Section~\ref{sec:appr-stat}. %
\end{rem}
\subsection{Common message encoding}
\label{sec:comm-mess-encod}
In addition to the polarization sets defined in~\eqref{eq:common_msg_sets_1}--\eqref{eq:common_msg_sets_6} we also define
\begin{align*}
  \mathcal{I}_{UY} & \triangleq {\mathcal{V}}_{U} \backslash {\mathcal{H}}_{U|Y},\\
\mathcal{I}_{UZ} & \triangleq {\mathcal{V}}_{U} \backslash {\mathcal{H}}_{U|Z},\\
\mathcal{A}_{UYZ}&\eqdef \text{a} \text{ subset\footnotemark\phantom{l}of $\mathcal{I}_{UZ} \backslash \mathcal{I}_{UY}$ with size $|  \mathcal{I}_{UY} \backslash \mathcal{I}_{UZ}|$.}
\end{align*}
\footnotetext{$\mathcal{A}_{UYZ}$ can be chosen as any subset of $\mathcal{I}_{UZ} \backslash \mathcal{I}_{UY}$, what matters is that $\mathcal{A}_{UYZ}$ is a subset of $\mathcal{I}_{UZ} \backslash \mathcal{I}_{UY}$ and inherits its properties.}
Note that $\mathcal{A}_{UYZ}$ exists because
\begin{align*}
  |\mathcal{I}_{UZ} \backslash \mathcal{I}_{UY}| - | \mathcal{I}_{UY}
  \backslash \mathcal{I}_{UZ}| = |\mathcal{I}_{UZ} |- |
  \mathcal{I}_{UY}|,
\end{align*}
and since we have assumed $I(U;Y)\leq I(U;Z)$, one can show with Lemmas \ref{lemcard_1}, \ref{lemcard_2},  $$\lim_{N \to \infty} (|\mathcal{I}_{UZ} |- |
  \mathcal{I}_{UY}|)/N \geq 0.$$
The encoding procedure with chaining is summarized in Figure~\ref{fig_atilde}.\smallskip
\begin{figure*}
\centering
\includegraphics[width=16.5cm]{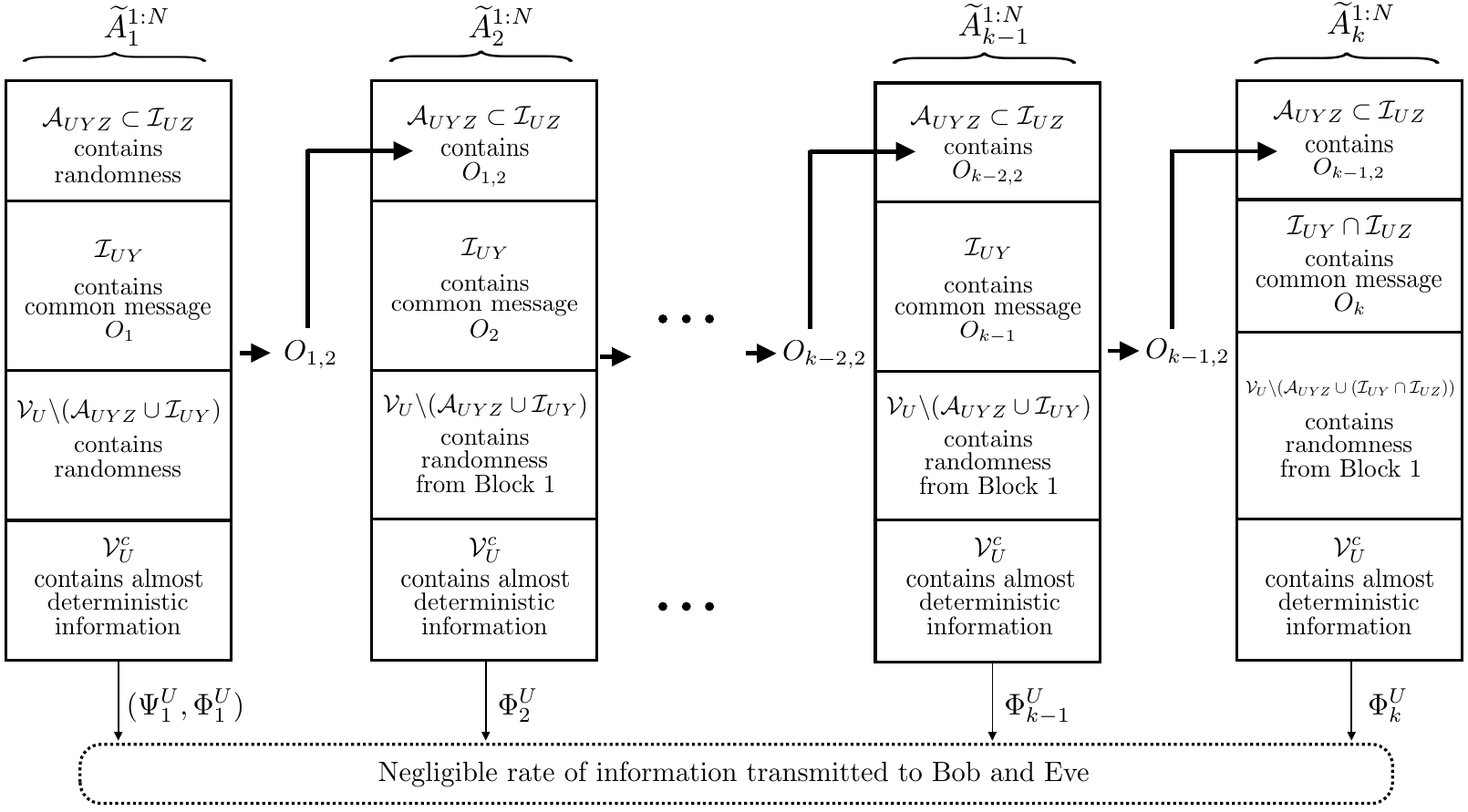}
  \caption{Chaining for the encoding of the $\widetilde{A}_i^{1:N}$'s, which corresponds to the encoding of the common messages. In Block $i \in \llbracket 1, k-1\rrbracket$, $\widetilde{A}_i^{1:N}$ is constructed from the common message $O_i$, the subsequence $O_{i-1,2} \triangleq \widetilde{A}_{i-1}^{1:N}[\mathcal{I}_{UY} \backslash \mathcal{I}_{UZ}]$ of the common message $O_{i-1}$, and part of the randomness $\Psi_1^U \triangleq \widetilde{A}_{1}^{1:N}[{\mathcal{V}}_{U} \backslash \mathcal{I}_{UY}  ]$ repeated from Block $1$. The remaining symbols of $\widetilde{A}_i^{1:N}$ are almost deterministic given $(O_i,O_{i-1,2},\Psi_1^U)$. Note that Block $k$ contains a smaller common message $O_k$ -- see the decoding scheme for more details. Finally, for all $i \in \llbracket 1, k\rrbracket$, $\Phi_i^U \triangleq \widetilde{A}_i^{1:N}[ ({\mathcal{H}}_{U|Y} \cup {\mathcal{H}}_{U|Z}) \cap {\mathcal{V}}_{U}^c]$, which is non-uniform and has negligible rate, is transmitted separately to Bob and Eve. $\Psi_1^U$ is also transmitted separately to Bob and Eve -- note that the rate of this transmission vanishes to zero as the number of blocks $k$ increases.}
  \label{fig_atilde}
\end{figure*}

In {Block ${1}$}, the encoder forms $\widetilde{U}_1^{1:N}$ as follows. Let ${O}_1$ be a vector of $|\mathcal{I}_{UY}|$ uniformly distributed $q^{(U)}$-ary symbols that represents  the common message to be reconstructed by Bob and Eve. Upon observing a realization $o_1$, the encoder samples $\widetilde{a}_1^{1:N}$ from the distribution $\widetilde{p}_{A_1^{1:N}}$ defined as 
\begin{multline} \label{eq_sim_A_1}
\widetilde{p}_{{A}_1^j|{A}_1^{1:j-1}} ({a}_1^j|{a}_1^{1:j-1}) \\ \triangleq
\begin{cases}
  \mathds{1} \left\{ a_1^j  = {o_1^j} \right\} & \text{if } j \in \mathcal{I}_{UY}\\
  1/q^{(U)} & \text{if }j \in {\mathcal{V}}_{U} \backslash \mathcal{I}_{UY}\\
  {p}_{A^j|A^{1:j-1}} (a_1^j|a_1^{1:j-1}) & \text{if }j\in {\mathcal{V}}_{U}^c
 \end{cases},
\end{multline}
where the components of $o_1$ have been indexed by the set of indices $\mathcal{I}_{UY}$ for convenience, so that $${O}_1 = \widetilde{A}_1^{1:N} [\mathcal{I}_{UY}].$$ The random $q^{(U)}$-ary symbols that identify the codebook and that are required to reconstruct $\widetilde{A}_1^{1:N}$ are $\widetilde{A}_1^{1:N}[ \mathcal{H}_{U|Z}]$ for Eve and  $\widetilde{A}_1^{1:N}[ \mathcal{H}_{U|Y}]$ for Bob. 
Moreover, we define
\begin{align*}
  \Psi^{U}_1 &\triangleq \widetilde{A}_1^{1:N}[ {\mathcal{V}}_{U} \backslash \mathcal{I}_{UY} ]=\widetilde{A}_1^{1:N}[ {\mathcal{V}}_{U} \cap\calH_{U|Y} ],\\
  \Phi^U_1 &\triangleq \widetilde{A}_1^{1:N}[ ({\mathcal{H}}_{U|Y} \cup {\mathcal{H}}_{U|Z}) \backslash {\mathcal{V}}_{U}].
\end{align*}
Both $\Psi^{U}_1$ and $\Phi^{U}_1$ are publicly transmitted to both Bob and Eve. Note that, unlike in the random binning proof, the use of polarization forces us to distinguish the part $\Psi^{U}_1$ that is nearly uniform from the part $\Phi^{U}_1$ that is not. We show later that the rate cost of this additional transmission is negligible. We also write 
$$O_{1} \triangleq [{O}_{1,1}, {O}_{1,2}],$$ where  
\begin{align*}
	{O}_{1,1} &\triangleq \widetilde{A}_1^{1:N} [\mathcal{I}_{UY} \cap \mathcal{I}_{UZ}], \\
	{O}_{1,2} &\triangleq \widetilde{A}_1^{1:N} [\mathcal{I}_{UY} \backslash \mathcal{I}_{UZ} ].
\end{align*} We will retransmit ${O}_{1,2}$ in the next block. Finally, we compute $$\widetilde{U}_1^{1:N} \triangleq \widetilde{A}_1^{1:N} G_n.$$

In {Block} ${i \in \llbracket 2,k-1\rrbracket}$, the encoder forms $\widetilde{U}_i^{1:N}$ as follows. Let $O_i$ be a vector of $|\mathcal{I}_{UY}|$ uniformly distributed $q^{(U)}$-ary symbols representing the common message in that block. Upon observing the realization ${o}_i$ and knowing $o_{i-1}$, the encoder draws $\widetilde{a}_i^{1:N}$ from the distribution $\widetilde{p}_{A_i^{1:N}}$ defined as follows.
\begin{align} 
&\widetilde{p}_{{A}_i^j|{A}_i^{1:j-1}} ({a}_i^j|{a}_i^{1:j-1}) \nonumber \\ & \phantom{l} \triangleq
\begin{cases}
  \mathds{1} \left\{ a_i^j  = {o_i^j} \right\} & \text{if } j \in \mathcal{I}_{UY}\\
  \mathds{1} \left\{ a_i^j  = {o_{i-1,2}^j} \right\} & \text{if } j \in \mathcal{A}_{UYZ} \\
  \mathds{1} \left\{ a_i^j  = (\psi^{U}_1)^j\right\} & \text{if } j \in {\mathcal{V}}_{U} \backslash (\mathcal{I}_{UY} \cup \mathcal{A}_{UYZ})  \\
    {p}_{A^j|A^{1:j-1}} (a_i^j|a_i^{1:j-1}) & \text{if }j\in {\mathcal{V}}_{U}^c
 \end{cases}, \label{eq_sim_A_i}
\end{align}
where the components of $o_i$, ${o_{i-1,2}}$, and $\psi^{U}_1$, have been indexed by the set of indices $\mathcal{I}_{UY}$, $\mathcal{A}_{UYZ}$, and $\mathcal{V}_{U} \backslash (\mathcal{I}_{UY} \cup \mathcal{A}_{UYZ})$, respectively. 
Consequently, note that $${O}_i = \widetilde{A}_i^{1:N} [\mathcal{I}_{UY}] \text{ and }O_{i-1,2} = \widetilde{A}_i^{1:N} [\mathcal{A}_{UYZ}].$$ The random $q^{(U)}$-ary symbols that identify the codebook and that are required to reconstruct $\widetilde{A}_i^{1:N}$ are $\widetilde{A}_i^{1:N} [\mathcal{H}_{U|Y}]$ for Bob and  $\widetilde{A}_i^{1:N}[ \mathcal{H}_{U|Z}]$ for Eve. %
We define
\begin{align*}
  \Psi^{U}_i &\triangleq \widetilde{A}_i^{1:N}[{\mathcal{V}}_{U} \backslash (\mathcal{I}_{UY} \cup \mathcal{A}_{UYZ})],\\
  \Phi^{U}_i &\triangleq \widetilde{A}_i^{1:N}[ ({\mathcal{H}}_{U|Y} \cup {\mathcal{H}}_{U|Z}) \backslash {\mathcal{V}}_{U}].
\end{align*}
Note that the $q^{(U)}$-ary symbols in $\Psi^{U}_i$ are reusing some of the $q^{(U)}$-ary symbols in $\Psi^{U}_1$; however, it is necessary to make the $q^{(U)}$-ary symbols $\Phi^{U}_i$ available to both Bob and Eve, to enable the reconstruction of $O_i$ -- See Remark \ref{remsubt}.i. We show later that this entails a negligible rate cost. Finally, we write \begin{align*}{O}_{i} \triangleq [{O}_{i,1}, {O}_{i,2}],\end{align*} where  
\begin{align*}
{O}_{i,1}\triangleq \widetilde{A}_i^{1:N} [\mathcal{I}_{UY} \cap \mathcal{I}_{UZ}], \\
{O}_{i,2}\triangleq\widetilde{A}_i^{1:N}  [\mathcal{I}_{UY} \backslash \mathcal{I}_{UZ}],
\end{align*} and we retransmit ${O}_{i,2}$ in the next block. We finally compute $$\widetilde{U}_i^{1:N} \triangleq \widetilde{A}_i^{1:N} G_n.$$

Finally, the encoder forms $\widetilde{U}_k^{1:N}$ in {Block} ${k}$, as follows. Let $O_k$ be a vector of $|\mathcal{I}_{UY} \cap \mathcal{I}_{UZ}|$ uniformly distributed $q^{(U)}$-ary symbols representing the common message in that block. Given realizations ${o}_k$ and $o_{k-1}$, the encoder samples $\widetilde{a}_k^{1:N}$ from the distribution $\widetilde{p}_{A_k^{1:N}}$ defined as follows.
\begin{align} \label{eq_sim_A_k}
&\widetilde{p}_{{A}_k^j|{A}_k^{1:j-1}} ({a}_k^j|{a}_k^{1:j-1}) \nonumber \\& \triangleq
\begin{cases}
  \mathds{1} \left\{ a_k^j  = {o_k^j} \right\}  &\text{if } j \in \mathcal{I}_{UY} \cap \mathcal{I}_{UZ}\\
    \mathds{1} \left\{ a_k^j  = {o_{k-1,2}^j} \right\} &\text{if } j \in \mathcal{A}_{UYZ}  \\
  \mathds{1} \left\{ a_k^j  = (\psi^{U}_1)^j\right\} &\!\!\!\!\!\!\!\!\! \text{if }j \in {\mathcal{V}}_{U} \backslash ( \mathcal{A}_{UYZ} \cup(\mathcal{I}_{UY} \cap \mathcal{I}_{UZ}) )  \\
  {p}_{A^j|A^{1:j-1}} (a_k^j|a_k^{1:j-1})  &\text{if }j\in {\mathcal{V}}_{U}^c
 \end{cases}
\end{align}
where the components of $o_k$, ${o_{k-1,2}}$, and $\psi^{U}_1$ have been indexed by the set of indices $\mathcal{I}_{UY} \cap \mathcal{I}_{UZ}$, $\mathcal{A}_{UYZ}$, and ${\mathcal{V}}_{U} \backslash ( \mathcal{A}_{UYZ} \cup(\mathcal{I}_{UY} \cap \mathcal{I}_{UZ}) )$, respectively. Consequently, 
$${O}_k = \widetilde{A}_k^{1:N} [\mathcal{I}_{UY} \cap \mathcal{I}_{UZ}], \text{ }O_{k-1,2} = \widetilde{A}_k^{1:N} [\mathcal{A}_{UYZ}].$$
The random $q^{(U)}$-ary symbols that identify the codebook and that are required to reconstruct $\widetilde{A}_k^{1:N}$ are $\widetilde{A}_k^{1:N} [\mathcal{H}_{U|Y}]$ for Bob and  $\widetilde{A}_k^{1:N}[ \mathcal{H}_{U|Z}]$ for Eve. %
We define
\begin{align*}
  \Psi^{U}_k &\triangleq \widetilde{A}_k^{1:N}[{\mathcal{V}}_{U} \backslash ( \mathcal{A}_{UYZ} \cup(\mathcal{I}_{UY} \cap \mathcal{I}_{UZ}) )],\\
  \Phi^{U}_k &\triangleq \widetilde{A}_k^{1:N}[({\mathcal{H}}_{U|Y} \cup {\mathcal{H}}_{U|Z}) \backslash {\mathcal{V}}_{U}],
\end{align*}
and note that $\Psi^{U}_k$ merely reuses some of the $q^{(U)}$-ary symbols of $\Psi^{U}_1$. $\Phi^{U}_k$ is made available to both Bob and Eve to help them reconstruct $O_k$, but this incurs a negligible rate cost. We finally compute $$\widetilde{U}_k^{1:N} \triangleq \widetilde{A}_k^{1:N} G_n.$$

The public transmission of $(\Psi^{U}_1,\Phi^{U}_{1:k})$ to perform the reconstruction of the common message is taken into account in the secrecy analysis in Section~\ref{sec:analys-polar-coding}.

\subsection{Secret and private message encoding}
\label{sec:secr-priv-mess}
In addition to the polarization set defined in~\eqref{eq:private_msg_sets_1}--\eqref{eq:private_msg_sets_5}, we also define
\begin{align*}
  \mathcal{B}_{V|UY}&\eqdef \text{a subset\footnotemark\phantom{l}of $\mathcal{V}_{V|UZ}$ with size $| {\mathcal{H}}_{V|UY}\cap {\mathcal{V}}_{V|U}|$}\\
  \mathcal{M}_{UVZ}&\triangleq \mathcal{V}_{V|U}\backslash \mathcal{V}_{V|UZ}.
\end{align*}
\footnotetext{$\mathcal{B}_{V|UY}$ can be chosen as any subset of $\mathcal{V}_{V|UZ}$, what matters is that $\mathcal{B}_{V|UY}$ is a subset of $\mathcal{V}_{V|UZ}$ and inherits its properties.}
The encoding procedure with chaining is summarized in Fig.~\ref{fig_btilde}.
\begin{figure*} 
\centering
  \includegraphics[width=16.cm]{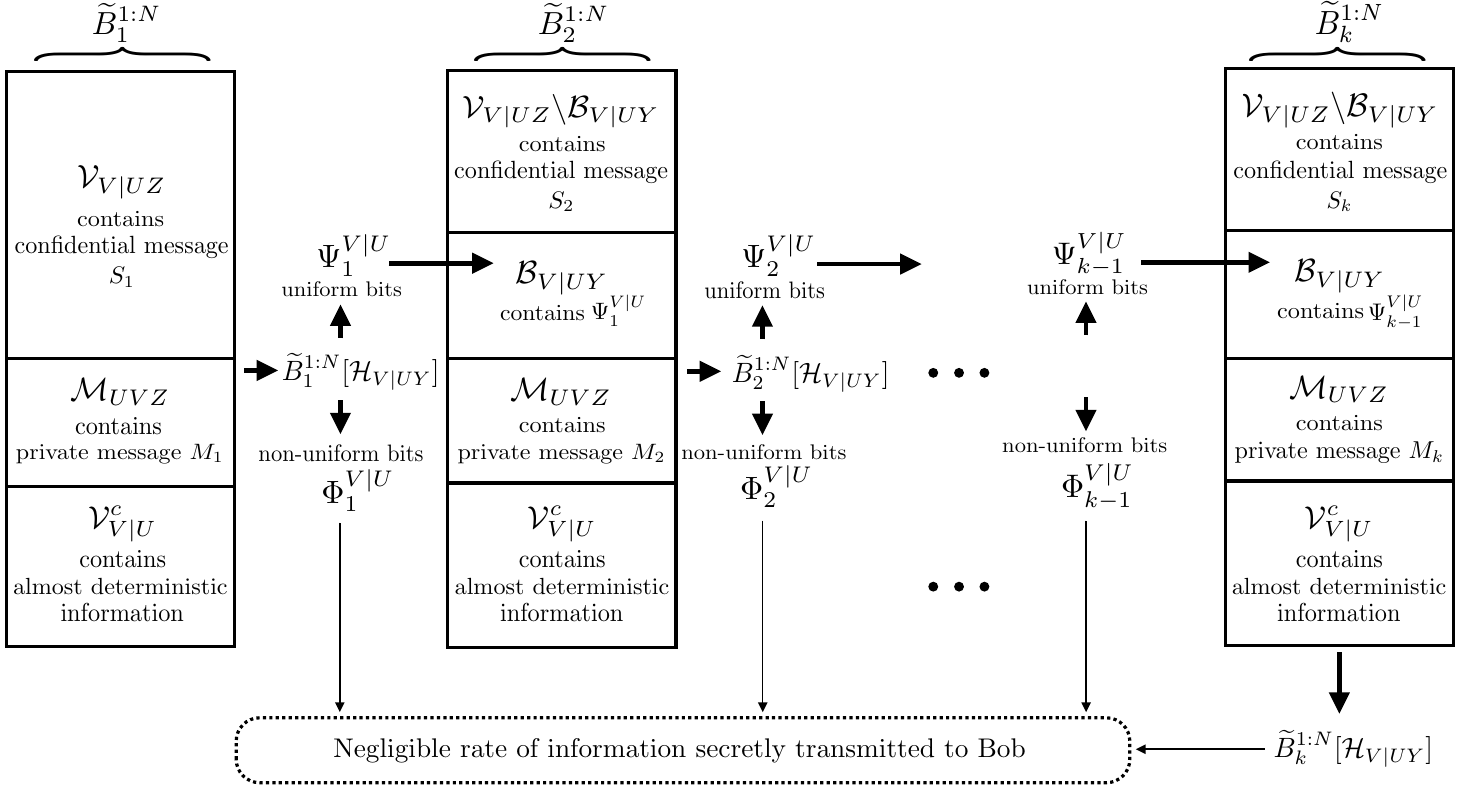}
  \caption{Chaining for the encoding of the $\widetilde{B}_i^{1:N}$'s, which corresponds to the encoding of the private and confidential messages. In Block $i\in \llbracket 1, k\rrbracket$, $\widetilde{B}_i^{1:N}$ is constructed from the confidential message $S_i$, the private message $M_i$, and the subsequence $\Psi^{V|U}_{i-1}$ of the previous block $\widetilde{B}_{i-1}^{1:N}$. The remaining symbols of $\widetilde{B}_i^{1:N}$ are almost deterministic given $(S_i,M_i,\Psi^{V|U}_{i-1})$.
   Note that $(\Psi_i^{V|U},\Phi_i^{V|U})$ is the information necessary to the legitimate receiver to recover~$\widetilde{B}_i^{1:N}$. Note also that $\Psi_i^{V|U}$ is uniform and repeated in Block $i+1$, whereas $\Phi_i^{V|U}$, whose rate is negligible, is non-uniform and secretly transmitted to the legitimate receiver with a one-time pad. Finally, $\widetilde{B}_k^{1:N}[\mathcal{H}_{V|UY}]$ is also secretly transmitted to the legitimate receiver with a one-time pad, and the rate of this transmission vanishes to zero as the number of blocks $k$ increases.}
  \label{fig_btilde}
\end{figure*}

In Block ${1}$, the encoder forms $\widetilde{V}_1^{1:N}$ as follows. Let ${S}_1$ be a vector of $|\mathcal{V}_{V|UZ}|$ uniformly distributed $q^{(V)}$-ary symbols representing the secret message and let ${M}_1$ be a vector of $|\mathcal{M}_{UVZ}|$ uniformly distributed $q^{(V)}$-ary symbols representing the private message to be reconstructed by Bob. Given a confidential message ${s}_1$, a private message $m_1$, and $\widetilde{u}_1^{1:N}$ resulting from the encoding of the common message, the encoder samples $\widetilde{b}_1^{1:N}$ from the distribution $\widetilde{p}_{B_1^{1:N}}$ defined as follows.
\begin{align} 
&\widetilde{p}_{{B}_1^j|{B}_1^{1:j-1}U_1^{1:N}} ({b}_1^j|{b}_1^{1:j-1} \widetilde{u}_1^{1:N}) \nonumber \\
& \triangleq 
\begin{cases}
   \mathds{1} \left\{ b_1^j = {s}_1^j \right\} & \text{if } j \in \mathcal{V}_{V|UZ}\\
  \mathds{1} \left\{ b_1^j = {m}_1^j \right\} & \text{if }  j \in \mathcal{M}_{UVZ}\\
  {p}_{B^j|B^{1:j-1}U^{1:N}} (b_1^j|b_1^{1:j-1} \widetilde{u}_1^{1:N}) & \text{if }j\in {\mathcal{V}}_{V|U}^c
 \end{cases}, \label{eq_sim_Bv_1}
\end{align}
where the components of $s_1$ and $m_1$ have been indexed by the set of indices $\mathcal{V}_{V|UZ}$  and $\mathcal{M}_{UVZ}$, respectively. Consequently, note that 
\begin{align*}
	S_1 & = \widetilde{B}_1^{1:N} [\mathcal{V}_{V|UZ}],\\
	M_1 & = \widetilde{B}_1^{1:N} [\mathcal{M}_{UVZ}].
\end{align*}
 The random $q^{(V)}$-ary symbols that identify the codebook required for reconstruction are those in positions $\calH_{V|UY}$, which we split as
 \begin{align*}
 &\Psi^{V|U}_1 \triangleq  \widetilde{B}_1^{1:N}[ {\mathcal{H}}_{V|UY}\cap {\mathcal{V}}_{V|U}] ,\\
 &\Phi^{V|U}_1 \triangleq \widetilde{B}_1^{1:N}[{\mathcal{H}}_{V|UY} \cap {\mathcal{V}}_{V|U}^c].
 \end{align*}
Note that $\Psi^{V|U}_1$ is uniformly distributed but $\Phi^{V|U}_1$ is not. Consequently, we may reuse $\Psi^{V|U}_1$ in the next block but we cannot reuse $\Phi^{V|U}_1$. We instead share $\Phi^{V|U}_1$ secretly between Alice and Bob and we show later that this may be accomplished with negligible rate cost. Finally, define $$\widetilde{V}_1^{1:N} \triangleq \widetilde{B}_1^{1:N} G_n.$$%

In Block ${i \in \llbracket 2,k\rrbracket}$, the encoder forms $\widetilde{V}_i^{1:N}$ as follows. Let $S_i$ be a vector of $|\mathcal{V}_{V|UZ}  \backslash \mathcal{B}_{V|UY}|$ uniformly distributed $q^{(V)}$-ary symbols and $M_i$ be a vector of $|\mathcal{M}_{UVZ}|$ uniformly distributed $q^{(V)}$-ary symbols  that represent the secret and private message in Block $i$, respectively. Given a private message $m_i$,  a confidential message $s_i$, $\psi^{V|U}_{i-1}$, and $\widetilde{u}_{i}^{1:N}$ resulting from the encoding of the common message, the encoder draws $\widetilde{b}_i^{1:N}$ from the distribution $\widetilde{p}_{B_i^{1:N}}$ defined as follows.\\

\begin{align}
&\widetilde{p}_{{B}_i^j|{B}_i^{1:j-1} U_i^{1:N}} ({b}_i^j|{b}_i^{1:j-1}\widetilde{u}_{i}^{1:N}) \nonumber\\
& \triangleq
\begin{cases}
 \mathds{1} \left\{ b_i^j = s_i^j \right\} & \text{if } j \in \mathcal{V}_{V|UZ}  \backslash \mathcal{B}_{V|UY}\\
  \mathds{1} \left\{ b_i^j = \left(\psi^{V|U}_{i-1}\right)^j \right\} & \text{if } j \in  \mathcal{B}_{V|UY}\\
    \mathds{1} \left\{ b_i^j = m_{i}^j \right\} & \text{if } j \in \mathcal{M}_{UVZ} \\
  {p}_{B^j|B^{1:j-1}U^{1:N}} (b_i^j|b_i^{1:j-1}\widetilde{u}_{i}^{1:N}) & \text{if }j\in {\mathcal{V}}_{V|U}^c
 \end{cases},\label{defsimBi} 
\end{align}
where the components of $s_i$, $\psi^{V|U}_{i-1}$, and $m_i$ have been indexed by the set of indices $\mathcal{V}_{V|UZ}  \backslash \mathcal{B}_{V|UY}$,  $ \mathcal{B}_{V|UY}$, and $\mathcal{M}_{UVZ}$ respectively, so that
\begin{align*}
S_i & = \widetilde{B}_i^{1:N} [\mathcal{V}_{V|UZ} \backslash \mathcal{B}_{V|UY}], \\	
\Psi^{V|U}_{i-1} & = \widetilde{B}_i^{1:N} [ \mathcal{B}_{V|UY}], \\
M_i & = \widetilde{B}_i^{1:N} [\mathcal{M}_{UVZ}].
\end{align*}
 The random $q^{(V)}$-ary symbols that identify the codebook required for reconstruction are those in positions $\calH_{V|UY}$, which we split as
 \begin{align*}
 &\Psi^{V|U}_i \triangleq  \widetilde{B}_i^{1:N}[{\mathcal{H}}_{V|UY} \cap {\mathcal{V}}_{V|U}] ,\\
&\Phi^{V|U}_i \triangleq \widetilde{B}_i^{1:N}[{\mathcal{H}}_{V|UY} \cap {\mathcal{V}}_{V|U}^c].
 \end{align*}
Again, $\Psi^{V|U}_i$ is uniformly distributed but $\Phi^{V|U}_i$ is not, so that we reuse $\Psi^{V|U}_i$ in the next block but we share $\Phi^{V|U}_i$ securely between Alice and Bob. We show later that the cost of sharing $\Phi^{V|U}_i$ is negligible. We then define $$\widetilde{V}_i^{1:N} \triangleq \widetilde{B}_i^{1:N} G_n.$$

In Block $k$, Alice securely shares $\left(\Psi^{V|U}_{k},\Phi^{V|U}_{1:k}\right)$ with Bob as follows. Alice performs a modulo-$q^{(V)}$ addition between $\left(\Psi^{V|U}_{k},\Phi^{V|U}_{1:k}\right)$ and a secret seed, i.e., a uniform sequence of $q^{(V)}$-ary symbols privately shared with Bob. Alice sends the result, which is a uniform sequence of $q^{(V)}$-ary symbols, to Bob by means of a channel polar code \cite{Sasoglu09}.\footnote{Note that a basic construction that achieves the symmetric capacity of the channel is sufficient here, as the length of the sequence transmitted is negligible compared to the overall blocklength $kN$.} Although this transmission incurs a rate loss, the later vanishes to zero as the length of the transmission is negligible compared to the overall blocklength $kN$. This point is detailed in Section~\ref{sec:coding-rates}.

\begin{rem} \label{remsubt}
The encoding of the secret messages requires a small pre-shared seed between the legitimate users for the two following reasons.
\begin{enumerate}[(i)]
\item In Lemma \ref{lm:1}, one cannot replace $\mathcal{H}_{X|Y}$ by $$\calV_{X|Y}\eqdef\{i\in\llbracket 1,N \rrbracket:\avgH{U^i|U^{1:i-1}Y^N}>1-\delta_N\},$$ i.e., $U^{1:N}$ \emph{cannot} be losslessly reconstructed from $U^{1:N}[\mathcal{V}_{X|Y}]$ and $Y^{1:N}$, although $|\mathcal{H}_{X|Y} | -| \mathcal{V}_{X|Y}|=o(N)  \text{ \cite[Lemma 1]{Chou14rev}}.$ This results from the trade-off between lossless source coding and the intrinsic randomness problem~\cite{Han05,Hayashi08,Chou13}. This translates in our coding scheme by the partition of $\widetilde{B}_i^{1:N}[\calH_{V|Y}]$ into $\Psi^{V}_i$ and $\Phi^{V}_i$, $i \in \llbracket 1,k\rrbracket$, where the non-uniform part $\Phi^{V}_i$ is secretly transmitted  from Alice to Bob thanks to a small pre-shared secret seed.
\item To deal with unaligned indices due to the potentially non-degraded channels, chaining also requires to secretly transmit $\Psi_k^V$ with a pre-shared secret seed in the last encoding block. 
\end{enumerate} 
\end{rem} 
\subsection{Channel prefixing}
\label{sec:channel-prefixing}
The channel prefixing procedure with chaining is illustrated in Fig.~\ref{fig_ttilde}.
\begin{figure}
\centering
  \includegraphics[width=8.7cm]{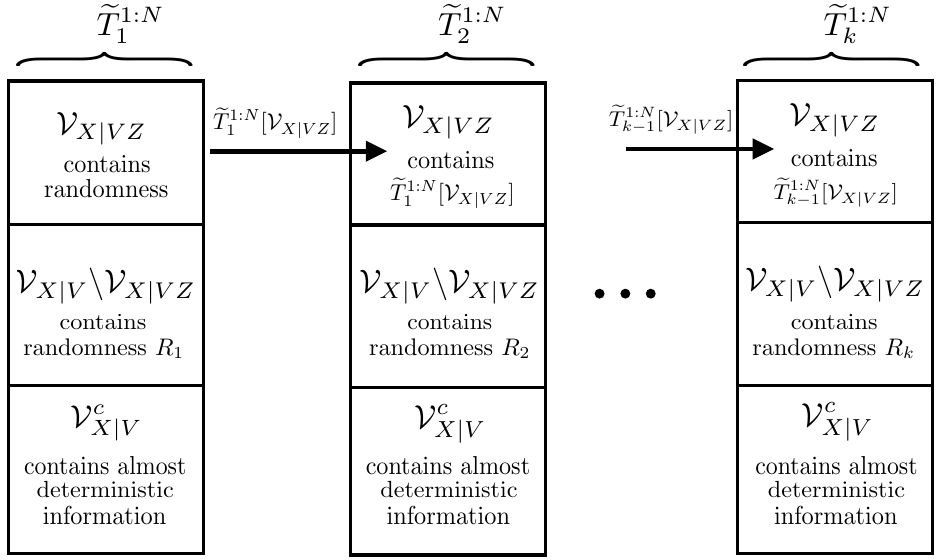}
  \caption{Chaining for the encoding of the $\widetilde{T}_i^{1:N}$'s, which corresponds to channel prefixing.  In Block $i\in \llbracket 1, k\rrbracket$, $\widetilde{T}_i^{1:N}$ is constructed from the randomness $R_i$, and the subsequence $\Psi^{X|V}_{i-1} \triangleq \widetilde{T}_{i-1}^{1:N}[\mathcal{V}_{X|VZ}] = \Psi^{X|V}_{1}$. The remaining symbols of $\widetilde{T}_i^{1:N}$ are almost deterministic given $(R_i,\Psi^{X|V}_{i-1})$.}
  \label{fig_ttilde}
\end{figure}
\smallskip

In Block ${1}$, the encoder forms $\widetilde{X}_1^{1:N}$ as follows. Let ${R}_1$ be a vector of $| \mathcal{V}_{X|V} \backslash \mathcal{V}_{X|VZ}|$ uniformly distributed $q^{(X)}$-ary symbols representing the randomness required for channel prefixing. Given a randomization sequence $r_1$ and $\widetilde{v}_1^{1:N}$ resulting from the encoding of secret and private messages, the encoder draws $\widetilde{t}_1^{1:N}$ from the distribution $\widetilde{p}_{T_1^{1:N}}$ defined as follows.
\begin{align}
& \widetilde{p}_{T_1^j|T_1^{1:j-1}V_1^{1:N}} (t_1^j|t_1^{1:j-1}\widetilde{v}_1^{1:N}) \nonumber\\ 
& \triangleq
\begin{cases}
  1/q^{(X)} & \text{if }j \in  \mathcal{V}_{X|VZ}\\
    \mathds{1} \left\{ t_1^j = r_{1}^j \right\}&\text{if }j\in \mathcal{V}_{X|V} \backslash \mathcal{V}_{X|VZ}\\
  {p}_{T^j|T^{1:j-1}V^{1:N}} (t_1^j|t_1^{1:j-1}\widetilde{v}_1^{1:N}) & \text{if }j\in \mathcal{V}_{X|V}^c
 \end{cases},  \label{defsimT_1} 
\end{align}
where the components of $r_1$ have been indexed by the set of indices $\mathcal{V}_{X|V} \backslash \mathcal{V}_{X|VZ}$, so that 
$$R_1 = \widetilde{T}_i^{1:N}[\mathcal{V}_{X|V} \backslash \mathcal{V}_{X|VZ}].$$
The random $q^{(X)}$-ary symbols that identify the codebook are those in position $\calV_{X|VZ}$, which we denote
\begin{align*}
  \Psi^{X|V}_1 \eqdef \widetilde{T}_1^{1:N}[\mathcal{V}_{X|VZ}].
\end{align*}
Finally, compute $$\widetilde{X}_1^{1:N} \triangleq \widetilde{T}_1^{1:N} G_n,$$ which is transmitted over the channel $W_{YZ|X}$. We note $Y_1^{1:N}$, $Z_1^{1:N}$ the corresponding channel outputs.
\smallskip

In {Block} ${i \in \llbracket 2,k\rrbracket}$, the encoder forms $\widetilde{X}_i^{1:N}$ as follows. Let ${R}_i$ be a vector of $| \mathcal{V}_{X|V} \backslash \mathcal{V}_{X|VZ}|$ uniformly distributed $q^{(X)}$-ary symbols representing the randomness required for channel prefixing in Block $i$. Given a randomization sequence $r_i$ and $\widetilde{v}_i^{1:N}$ resulting from the encoding of secret and private messages, the encoder draws $\widetilde{t}_i^{1:N}$ from the distribution $\widetilde{p}_{T_i^{1:N}}$ defined as follows.
\begin{align} 
&\widetilde{p}_{T^j_i|T_i^{1:j-1}V_i^{1:N}} (t_i^j|t_i^{1:j-1}\widetilde{v}_i^{1:N}) \nonumber \\ 
& \triangleq
\begin{cases}
\mathds{1} \left\{ t^j_i =  \widetilde{t}^j_{i-1} \right\} & \text{if } j\in \mathcal{V}_{X|VZ}  \\
\mathds{1} \left\{ t_i^j = r_{i}^j \right\} & \text{if }j \in \mathcal{V}_{X|V} \backslash \mathcal{V}_{X|VZ}\\
  {p}_{T^j|T^{1:j-1}V^{1:N}} (t_i^j|t_i^{1:j-1}\widetilde{v}_i^{1:N}) & \text{if }j \in \mathcal{V}_{X|V}^c
 \end{cases}, \label{defsimTi}
\end{align}
where the components of $r_i$ have been indexed by the set of indices $\mathcal{V}_{X|V} \backslash \mathcal{V}_{X|VZ}$, so that $$R_i = \widetilde{T}_i^{1:N}[\mathcal{V}_{X|V} \backslash \mathcal{V}_{X|VZ}].$$
Note that the random $q^{(X)}$-ary symbols describing the codebook are $$  \Psi^{X|V}_{i} \eqdef \widetilde{T}_i^{1:N}[\mathcal{V}_{X|VZ}],$$ and are reused from the previous block. Finally, define $$\widetilde{X}_i^{1:N} \triangleq \widetilde{T}_i^{1:N} G_n$$ and transmit it over the channel $W_{YZ|X}$. We denote the corresponding channel outputs by $Y_i^{1:N}$ and $Z_i^{1:N}$.

\subsection{Decoding}
\label{sec:decoding}

Reconstruction of the common message by Bob and Eve follows the idea of \cite{Mondelli14b}, i.e., backward decoding for Eve and forward decoding for Bob. More specifically, the decoding procedure is as follows.

\noindent \textbf{Reconstruction of the common message by Bob.} Bob forms the estimate $\widehat{A}_{1:k}^{1:N}$ of $\widetilde{A}_{1:k}^{1:N}$ as follows. In Block 1, Bob knows $(\Psi^U_1,\Phi^U_1)$, which contains all the $q^{(U)}$-ary symbols $\widetilde{A}_1^{1:N} [  \mathcal{H}_{U|Y}]$ by construction. Bob runs the successive cancellation decoder for source coding with side information of~\cite{Arikan10} using $Y_1^{1:N}$ and $\widetilde{A}_1^{1:N} [  \mathcal{H}_{U|Y}]$ to form ${{\widehat{A}}}_{1}^{1:N}$, an estimate of ${{\widetilde{A}}}_{1}^{1:N}$. In Block $i \in \llbracket 2, k\rrbracket$, Bob estimates $\widetilde{A}_i^{1:N} [  \mathcal{H}_{U|Y}]$ with  $(\Psi^U_1,\widehat{A}_{i-1}^{1:N}[ \mathcal{I}_{UY} \backslash \mathcal{I}_{UZ}], \Phi^U_i)$,\footnote{Observe that $ [{\mathcal{V}}_{U} \backslash (\mathcal{I}_{UY} \cup \mathcal{A}_{UYZ})] \cup \mathcal{A}_{UYZ} \cup [({\mathcal{H}}_{U|Y} \cup {\mathcal{H}}_{U|Z}) \backslash {\mathcal{V}}_{U}] \supset \mathcal{H}_{U|Y}$.} and uses this estimate along with $Y_i^{1:N}$ to run the successive cancellation decoder for source coding with side information to form ${{\widehat{A}}}_{i}^{1:N}$, an estimate of ${{\widetilde{A}}}_{i}^{1:N}$.\smallskip%

\noindent\textbf{Reconstruction of the common message by Eve.} Eve forms the estimate $\widehat{\widehat{A}}_{1:k}^{1:N}$ of $\widetilde{A}_{1:k}^{1:N}$ starting from Block $k$ and going backwards as follows. In Block $k$, Eve knows $(\Psi^U_k,\Phi^U_k)$, which contains all the $q^{(U)}$-ary symbols in $ \widetilde{A}_k^{1:N} [  \mathcal{H}_{U|Z}]$ by construction.\footnote{Using that $\mathcal{A}_{UYZ}$ is a subset of $\mathcal{I}_{UZ} \backslash \mathcal{I}_{UY}$, observe that $ [{\mathcal{V}}_{U} \backslash ( \mathcal{A}_{UYZ} \cup(\mathcal{I}_{UY} \cap \mathcal{I}_{UZ}) ) ] \cup [({\mathcal{H}}_{U|Y} \cup {\mathcal{H}}_{U|Z}) \backslash {\mathcal{V}}_{U}] \supset \mathcal{H}_{U|Z}$. } Eve runs the successive cancellation decoder for source coding with side information using $Z_k^{1:N}$ and $\widetilde{A}_k^{1:N} [  \mathcal{H}_{U|Z}]$ to form $\smash{\widehat{\widehat{A}}}_{k}^{1:N}$, an estimate of ${{\widetilde{A}}}_{k}^{1:N}$. For $i \in \llbracket 1, k-1  \rrbracket$, Eve estimates $\widetilde{A}_{k-i}^{1:N}[ \mathcal{H}_{U|Z} ] $ with $(\Psi^U_1,\smash{\widehat{\widehat{A}}}_{k-i+1}^{1:N}[ \mathcal{A}_{UYZ} ], \Phi^U_{k-i} )$,\footnote{Using that $\mathcal{A}_{UYZ}$ is a subset of $\mathcal{I}_{UZ} \backslash \mathcal{I}_{UY}$, observe that $ [{\mathcal{V}}_{U} \backslash (\mathcal{I}_{UY} \cup \mathcal{A}_{UYZ})] \cup [\mathcal{I}_{UY} \backslash \mathcal{I}_{UZ}] \cup [({\mathcal{H}}_{U|Y} \cup {\mathcal{H}}_{U|Z}) \backslash {\mathcal{V}}_{U}] \supset \mathcal{H}_{U|Z}$. } and uses this estimate along with $Z_{k-i}^{1:N}$ to run the successive cancellation decoder for source coding with side information to form $\smash{\widehat{\widehat{A}}}_{k-i}^{1:N}$, an estimate of ${{\widetilde{A}}}_{k-i}^{1:N}$.\smallskip

\noindent\textbf{Reconstruction of the private and confidential messages by Bob.} Bob forms the estimate $\widehat{B}_{1:k}^{1:N}$ of $\widetilde{B}_{1:k}^{1:N}$ as follows starting with Block $k$. In Block $k$, given $(\Psi^{V|U}_{k},\Phi^{V|U}_{k} , Y_{k}^{1:N}, \widehat{U}_{k}^{1:N}) $, Bob forms $\widehat{B}_k^{1:N}$, an estimate of $\widetilde{B}_k^{1:N}$, with the successive cancellation decoder for source coding with side information. From $\widehat{B}_k^{1:N}$, an estimate $\widehat{\Psi}^{V|U}_{k-1}\eqdef\widehat{B}_k^{1:N} [{\mathcal{V}}_{V|UY}] $ of $\Psi^{V|U}_{k-1}$ is formed. For $i \in \llbracket 1,k-1 \rrbracket$, given $(\widehat{\Psi}^{V|U}_{k-i}, \Phi^{V|U}_{k-i} , Y_{k-i}^{1:N}, \widehat{U}_{k-i}^{1:N}) $, Bob forms $\widehat{B}_{k-i}^{1:N}$, an estimate of $\widetilde{B}_{k-i}^{1:N}$, with the successive cancellation decoder for source coding with side information. From $\widehat{B}_{k-i}^{1:N}$, an estimate of ${\Psi}^{V|U}_{k-i-1}$ is formed. Once all the estimates $\widehat{B}_{1:k}^{1:N}$ have been formed, Bob forms the estimates $\widehat{S}_{1:k}$ and $\widehat{M}_{1:k}$ of $S_{1:k}$ and $M_{1:k}$, respectively.

\section{Analysis of the Polar coding scheme}
\label{sec:analys-polar-coding}

We now analyze in details the characteristics and performances of the polar coding scheme described in Section~\ref{sec:polar-coding-schem}. Specifically, we show the following.

\begin{thm} \label{Thprep}
Consider a discrete memoryless broadcast channel $(\mathcal{X}, p_{YZ|X}, \mathcal{Y},\mathcal{Z})$. The coding scheme of Section~\ref{Sec_CS}, which operates over $k$ encoding blocks of length $N$ and whose complexity is $O(kN \log N)$ achieves the region $\mathcal{R}_{\textup{BCC}}$.
\end{thm}

The result of Theorem~\ref{Thprep}, follows in four steps. First, we show that the polar coding scheme of Section~\ref{sec:polar-coding-schem} approximates the statistics of the original \ac{DMS} $(\calU\times\calV\times\calX\times\calY\times\calZ,p_{UVXYZ})$ from which the polarization sets were defined. Second, we show that the various messages rates are indeed those in $\mathcal{R}_{\textup{BCC}}$. Third, we show that the probability of decoding error vanishes with the block length. Finally, we show that the information leakage vanishes with the block length. 

\subsection{Approximation of original \ac{DMS} statistics}
\label{sec:appr-stat}

Recall that the vectors $\widetilde{A}_i^{1:N}$, $\widetilde{B}_i^{1:N}$, $\widetilde{V}_i^{1:N}$, and $\widetilde{X}_i^{1:N}$, generated in Block $i \in \llbracket 1,k\rrbracket$ do not have the exact joint distribution of the vectors ${A}^{1:N}$, ${B}^{1:N}$, ${V}^{1:N}$, and ${X}^{1:N}$, induced by the source polarization of the original \ac{DMS} $(\calU\times\calV\times\calX\times\calY\times\calZ,p_{UVXYZ})$. However, the following lemma shows that the joint distributions are close to one another, which is crucial for the subsequent reliability and secrecy analysis. 

\begin{lem} \label{lem_dist_A}
For $i \in \llbracket 1,k \rrbracket$, we have
\begin{align*}
\mathbb{V}(p_{A^{1:N}}, \widetilde{p}_{A_i^{1:N}}) 
& \leq \delta_N^{(U)}, \\
\mathbb{V}(p_{B^{1:N}U^{1:N}}, \widetilde{p}_{B_i^{1:N}U_i^{1:N}}) 
& \leq \delta_N^{(UV)},\\
\mathbb{V}(p_{X^{1:N}V^{1:N}}, \widetilde{p}_{X_i^{1:N}V_i^{1:N}}) 
& \leq \delta_N^{(XV)},
\end{align*}
where
\begin{align*}
\delta_N^{(U)}  & \triangleq \sqrt{2\log 2} \sqrt{ N \delta_N }, \\
\delta_N^{(UV)} & \triangleq 2\sqrt{\log 2} \sqrt{ N \delta_N }, \\
\delta_N^{(XV)} & \triangleq \sqrt{2\log 2} \sqrt{ 3 N \delta_N }.
\end{align*}

Combining the three previous inequalities, we obtain 
\begin{align*}
\mathbb{V}(p_{U^{1:N}V^{1:N}X^{1:N}Y^{1:N}Z^{1:N}}, \widetilde{p}_{U_i^{1:N}V_i^{1:N}X_i^{1:N}Y_i^{1:N}Z_i^{1:N}}) 
& \leq \delta_N^{(P)}.
\end{align*}
where $ \delta_N^{(P)} \triangleq \sqrt{2\log 2} \sqrt{N \delta_N}(2 \sqrt{ 2 } +  \sqrt{3}).$
\end{lem}

\begin{proof}
See Appendix \ref{App_lem_dist_A}.
\end{proof}

\subsection{Transmission rates}
\label{sec:coding-rates}
We now analyze the rate of common message, confidential message, private message, and randomization sequence, used at the encoder, as well as the different sum rates and the rate of additional information sent to Bob and Eve. We will use the following lemmas.  %
\begin{lem}[Adapted from {\cite[Theorem 3.5]{Sasoglu11}} ] \label{lemcard_1}
	Consider a source $(\mathcal{X} \mathcal{Y}, p_{XY})$ with $|\mathcal{X}| =q$, $q$ prime and  $\mathcal{Y}$ a countable alphabet. Define ${U}^{1:N} \triangleq {X}^{1:N} G_n$ and for $\delta_N \triangleq 2^{-N^{\beta}}$, $\beta < 1/2$,
	$$
	\mathcal{H}_{X|Y} \triangleq \{i \in \llbracket 1,N \rrbracket :  H(U^i|U^{1:i-1} Y^{1:N}) > \delta_N \}.
	$$
We have	$$
	\lim_{N\to \infty} \frac{| \mathcal{H}_{X|Y}|}{N} = H(X|Y).
	$$
\end{lem}

\begin{lem} \label{lemcard_2}
	Consider a source $(\mathcal{X} \mathcal{Y}, p_{XY})$ with $|\mathcal{X}| =q$, $q$ prime and  $\mathcal{Y}$ a countable alphabet. Define ${U}^{1:N} \triangleq {X}^{1:N} G_n$ and for $\delta_N \triangleq 2^{-N^{\beta}}$, $\beta < 1/2$,
	$$
	\mathcal{V}_{X|Y} \triangleq \{i \in \llbracket 1,N \rrbracket :  H(U^i|U^{1:i-1} Y^{1:N}) > \log_2 (q) - \delta_N \}.
	$$
	We have 
	$$
	\lim_{N\to \infty} \frac{| \mathcal{V}_{X|Y}|}{N} = H(X|Y).
	$$
\end{lem}

\begin{proof}
See Appendix \ref{App_card}.	
\end{proof}
\begin{rem} \label{remcard}
Although the case $q=2$ first appeared in \cite{Honda13} and \cite[Lemma 1]{Chou13b}, Lemma~\ref{lemcard_2} has not appeared anywhere to the best of our knowledge. 
	A weaker result has been shown in~\cite[Theorem 3.4]{Sasoglu11}, specifically, for all $\epsilon >0$,
	\begin{multline*}
	\lim_{N\to \infty} \frac{|\{i \in \llbracket 1,N \rrbracket :  H(U^i|U^{1:i-1} Y^{1:N}) > \log_2 (q)- \epsilon \}|}{N} \\= H(X|Y).
	\end{multline*}

\end{rem}

\noindent\textbf{Common message rate.} The overall rate $R_O$ of common information  transmitted satisfies  
\begin{align*}
R_O 
& = \frac{ (k-1)|\mathcal{I}_{UY}|+ |\mathcal{I}_{UY}\cap \mathcal{I}_{UZ}|} {kN} \\
& = \frac{|\mathcal{I}_{UY}|}{N} -  \frac{  |\mathcal{I}_{UY}\backslash  \mathcal{I}_{UZ}|} {kN} \\
& \geq \frac{|\mathcal{I}_{UY}|}{N} -  \frac{  |\mathcal{I}_{UY}|} {kN} \\
& \xrightarrow{N \to \infty} I(Y;U) -  \frac{  I(Y;U)} {k}\\
& \xrightarrow{k \to \infty} I(Y;U),
\end{align*}
where we have used Lemma \ref{lemcard_1} and Lemma \ref{lemcard_2}. Since we also have $R_O \leq \frac{|\mathcal{I}_{UY}|}{N} \xrightarrow{N \to \infty} I(Y;U)$, we conclude
\begin{align}
R_O \xrightarrow{N \to \infty, k \to \infty} I(Y;U). \label{eqsum1}
\end{align}

\noindent\textbf{Confidential message rate.} First, observe that
\begin{align*}
|\Psi^{V|U}_1|  
& = |{\mathcal{H}}_{V|UY} \cap {\mathcal{V}}_{V|U}| \\
& \leq  |\mathcal{H}_{V|UY} |, 
\end{align*}
and $|\Psi^{V|U}_1| \geq |\mathcal{V}_{V|UY}|$ because $\mathcal{V}_{V|UY} \subseteq {\mathcal{H}}_{V|UY}$ and $\mathcal{V}_{V|UY} \subseteq {\mathcal{V}}_{V|U}$. Hence, since $\lim_{N \to \infty} |\mathcal{V}_{V|UY} |/N =H(V|UY)$ by Lemma \ref{lemcard_2} and $\lim_{N \to \infty} |\mathcal{H}_{V|UY} |/N =H(V|UY)$ by Lemma \ref{lemcard_1}, we have 
$$
\lim_{N \to \infty}  \frac{|\Psi^{V|U}_1|}{N}  =H(V|UY).
$$

\noindent{}Then, the overall rate $R_S$ of secret information transmitted is  
\begin{align}
R_S =
& \frac{ |\mathcal{V}_{V|UZ}|+ (k-1) |\mathcal{V}_{V|UZ} \backslash \mathcal{B}_{V|UY}|} {kN} \nonumber \\ \nonumber
& = \frac{ |\mathcal{V}_{V|UZ}|+ (k-1) (|\mathcal{V}_{V|UZ} |- | \mathcal{B}_{V|UY}|)} {kN} \\ \nonumber
& =  \frac{ |\mathcal{V}_{V|UZ} |- | \mathcal{B}_{V|UY}|} {N} + \frac{ |\mathcal{B}_{V|UY}|}{kN} \\  \nonumber
&=  \frac{ |\mathcal{V}_{V|UZ} |- | \Psi^{V|U}_1|} {N}  + \frac{ |\Psi^{V|U}_1|}{kN} \\  \nonumber
& \xrightarrow{N \to \infty} I(V;Y|U) - I(V;Z|U)  + \frac{ H(V|UY)}{k}\\
& \xrightarrow{k \to \infty} I(V;Y|U) - I(V;Z|U) . \label{eqsum2}
\end{align}

\noindent\textbf{Private message rate.} The overall rate $R_M$ of private information  transmitted is  
\begin{align}
R_M
&= \frac{ k |\mathcal{M}_{UVZ}|} {kN}\nonumber\\ \nonumber
&= \frac{  |\mathcal{V}_{V|U} \backslash \mathcal{V}_{V|UZ}|} {N} \\ \nonumber
&= \frac{  |\mathcal{V}_{V|U} |-| \mathcal{V}_{V|UZ}|} {N} \\
& \xrightarrow{N \to \infty} I(V;Z|U) , \label{eqsum3}
\end{align}
where we have used Lemma \ref{lemcard_2}.

\noindent\textbf{Randomization rate.} The randomness used in the stochastic encoder includes the randomization sequence for channel prefixing, as well as the randomness required to identify the codebooks and run the successive cancellation encoding. Using~Lemma \ref{lemcard_2}, we find that the rate required to identify the codebook for the common message is 
\begin{align*}
\frac{ |\mathcal{V}_{U} \backslash \mathcal{I}_{UY}|} {kN} \leq   \frac{ |\mathcal{V}_{U} |} {kN} \xrightarrow{N \to \infty} \frac{H(U|Y)}{k}\xrightarrow{k \to \infty} 0.
 \end{align*} 
Similarly, the rate required to identify the codebook for the secret and private messages corresponds to the rate of $(\Psi^{V|U}_{k},\Phi^{V|U}_{k})$, which is transmitted to Bob to allow him to reconstruct $\widetilde{B}_{1:k}^{1:N}$,
\begin{align*}
 \frac{ |(\Psi^{V|U}_{k},\Phi^{V|U}_{k})|} {kN} 
&  = \frac{ |\widetilde{B}_{k}^{1:N}[\mathcal{H}_{V|UY}]|} {kN} \\
& \xrightarrow{N \to \infty} \frac{ H(V|UY)}{k} \\
 & \xrightarrow{k \to \infty} 0,
 \end{align*}
 where we have used Lemma \ref{lemcard_1}.

The randomization sequence rate used in channel prefixing~is
\begin{align*}
 & \frac{ |\mathcal{V}_{X|V}| + (k-1) |\mathcal{V}_{X|V} \backslash \mathcal{V}_{X|VZ}|} {kN} \\
  & = \frac{ |\mathcal{V}_{X|V} \backslash \mathcal{V}_{X|VZ}|} {N} + \frac{ |\mathcal{V}_{X|VZ}|} {kN} \\
  & = \frac{ |\mathcal{V}_{X|V}|- | \mathcal{V}_{X|VZ}|} {N} + \frac{ |\mathcal{V}_{X|VZ}|} {kN} \\
 & \xrightarrow{N \to \infty} I(X;Z|V) + \frac{ H(X|VZ)} {k}, \\
 & \xrightarrow{k \to \infty} I(X;Z|V), 
 \end{align*}
where we have used Lemma \ref{lemcard_2}. Finally, we justify that the rate of uniform randomness required for successive cancellation encoding in (\ref{eq_sim_A_1})--(\ref{defsimTi}) is negligible in Appendix \ref{App_lemrandA}.

Hence, the overall randomness rate $R_R$ used at the encoder is asymptotically  \begin{align}
R_R \xrightarrow{N \to \infty,k \to \infty} I(X;Z|V). \label{eqsum4}
\end{align}

\noindent\textbf{Sum rates}. By \eqref{eqsum3} and \eqref{eqsum4}, the sum of the private message rate $R_M$ and the randomness rate $R_R$ is asymptotically
\begin{align*}
&R_M +R_R \\
& \xrightarrow{N \to \infty, k \to \infty}  I(V;Z|U) + I(X;Z|V) \\
& \stackrel{(a)}{=} H(Z|U) - H(Z|UV) + H(Z|V) - H(Z|XV)\\
& = H(Z|U) - H(Z|XV)\\
& \stackrel{(b)}{=} H(Z|U) - H(Z|XU)\\
& = I(X;Z|U),
\end{align*}
where $(a)$ and $(b)$ hold by $U - V - X - Z$.

Moreover, by \eqref{eqsum1}, \eqref{eqsum2}, and \eqref{eqsum3}, the sum of the common message rate $R_O$, the private message rate $R_M$, and the confidential message rate $R_S$ is asymptotically
\begin{align*}
& R_O + R_M +R_S \xrightarrow{N \to \infty, k \to \infty}  I(Y;U) +  I(V;Y|U) .
\end{align*}

\noindent\textbf{Seed Rate.}
The rate of the secret sequence that must be shared between the legitimate users to initialize the coding scheme is
\begin{align*}
\frac{ |\Psi^{V|U}_k| + k|\Phi^{V|U}_1| }{kN}
& = \frac{ |\Psi^{V|U}_k| }{kN} + \frac{ |\Phi^{V|U}_1| }{N}\\
& \leq \frac{ |{\mathcal{H}}_{V|UY}| }{kN} + \frac{ |{\mathcal{H}}_{V|UY} \backslash {\mathcal{V}}_{V|UY}| }{N} \\
& \leq \frac{ |{\mathcal{H}}_{V|UY}| }{kN} + \frac{ |{\mathcal{H}}_{V|UY}| - |{\mathcal{V}}_{V|UY}| }{N} \\
& \xrightarrow{N \to \infty}  \frac{H(V|Y)}{k} \\
&\xrightarrow{k \to \infty} 0,
\end{align*}
where we have used Lemma \ref{lemcard_1} and Lemma \ref{lemcard_2}.

Moreover the rate of public communication from Alice to both Bob and Eve is
\begin{align*}
\frac{ |\Psi^{U}_1| + |\Phi^{U}_{1:k}| }{kN} 
& \leq \frac{ |\Psi^{U}_1| + k|{\mathcal{H}}_{U} \backslash {\mathcal{V}}_{U}| }{kN} \\
& = \frac{ |\mathcal{V}_{U} \backslash \mathcal{I}_{UY}| + k(|{\mathcal{H}}_{U} | - | {\mathcal{V}}_{U}| )}{kN} \\
& \leq \frac{ |\mathcal{H}_{U|Y}|  + k(|{\mathcal{H}}_{U} | - | {\mathcal{V}}_{U}|) }{kN} \\
& = \frac{ |\mathcal{H}_{U|Y}| }{kN} +  \frac{ |{\mathcal{H}}_{U} | - | {\mathcal{V}}_{U}| }{N} \\
& \xrightarrow{N \to \infty}  \frac{H(U|Y)}{k}  \\
&\xrightarrow{ k \to \infty} 0.
\end{align*}

\subsection{Average probability of error}
\label{sec:aver-prob-error}

We first show that Eve and Bob can reconstruct the common messages $O^{1:N}_{1:k}$ with small error probability. For $i \in \llbracket 1,k \rrbracket$,
consider an optimal coupling~\cite[Lemma 3.6]{Aldous83} between $\widetilde{p}_{U_i^{1:N}Y_i^{1:N}}$ and $p_{U^{1:N}Y^{1:N}}$ such that $$\mathbb{P} [\mathcal{E}_{UY,i}] = \mathbb{V}(\widetilde{p}_{U_i^{1:N}Y_i^{1:N}} ,p_{U^{1:N}Y^{1:N}}),$$ where $\mathcal{E}_{UY,i} \triangleq \{ (\widetilde{U}_i^{1:N}, \widetilde{Y}_i^{1:N}) \neq ({U}^{1:N} , {Y}^{1:N})\}$. Define also for $i \in \llbracket 2 , k \rrbracket$, $$\mathcal{E}_{i} \triangleq \{ \widehat{A}_{i-1}^{1:N} [\mathcal{I}_{UY} \backslash \mathcal{I}_{UZ}] \neq  \widetilde{A}_{i-1}^{1:N}[\mathcal{I}_{UY} \backslash \mathcal{I}_{UZ}]\}.$$ 
We have
\begin{align}
\mathbb{P}[ O_{i} \neq \widehat{O}_{i}] \nonumber  \nonumber
& \leq \mathbb{P}[ \widehat{U}^{1:N}_{i} \neq \widetilde{U}^{1:N}_{i}]\\ \nonumber
& = \mathbb{P}[ \widehat{U}^{1:N}_{i} \neq \widetilde{U}^{1:N}_{i} |\mathcal{E}_{UY,i}^c\cap \mathcal{E}_i^c] \mathbb{P}[ \mathcal{E}_{UY,i}^c\cap \mathcal{E}_i^c] \\ \nonumber
& \phantom{llm}+ \mathbb{P}[ \widehat{U}^{1:N}_{i} \neq \widetilde{U}^{1:N}_{i} |\mathcal{E}_{UY,i} \cup \mathcal{E}_i] \mathbb{P}[ \mathcal{E}_{UY,i}\cup \mathcal{E}_i]  \\ \nonumber
& \leq \mathbb{P}[ \widehat{U}^{1:N}_{i} \neq \widetilde{U}^{1:N}_{i}  |\mathcal{E}_{UY,i}^c \cap \mathcal{E}_i^c]  +  \mathbb{P}[ \mathcal{E}_{UY,i} \cup \mathcal{E}_i]  \\ \nonumber
& \stackrel{(a)}{\leq}  N \delta_N + \mathbb{P}[ \mathcal{E}_{UY,i}] +  \mathbb{P}[  \mathcal{E}_i]  \\ \nonumber
& \stackrel{(b)}{\leq}  N \delta_N + \delta_N^{(P)} +  \mathbb{P}[  \mathcal{E}_i]\\ \nonumber
& \leq   N \delta_N + \delta_N^{(P)} +  \mathbb{P}[ \widehat{U}^{1:N}_{i-1} \neq \widetilde{U}^{1:N}_{i-1}]\\ \nonumber
& \stackrel{(c)}{\leq}  (i-1)(N \delta_N + \delta_N^{(P)} )+  \mathbb{P}[ \widehat{U}^{1:N}_{1} \neq \widetilde{U}^{1:N}_{1}]  \\
& \stackrel{(d)}{\leq}  i(N \delta_N + \delta_N^{(P)} ), \label{eq_err_utilde}
\end{align}
where $(a)$ follows from the error probability of source coding with side information \cite{Arikan10} and the union bound, $(b)$ holds by the optimal coupling and Lemma~\ref{lem_dist_A}, $(c)$ holds by induction since we have shown that for any $i \in \llbracket 2, k\rrbracket$, $$ \mathbb{P}[ \widehat{U}^{1:N}_{i} \neq \widetilde{U}^{1:N}_{i}] \leq N \delta_N + \delta_N^{(P)} +  \mathbb{P}[ \widehat{U}^{1:N}_{i-1} \neq \widetilde{U}^{1:N}_{i-1}],$$ $(d)$ holds similarly to the previous inequalities. We thus have by the  union bound and (\ref{eq_err_utilde})
\begin{align*}
\mathbb{P}[ O^{1:N}_{1:k} \neq \widehat{O}^{1:N}_{1:k}] 
& \leq \sum_{i=1}^k \mathbb{P}[ O_{i} \neq \widehat{O}_{i}] \\
& \leq \frac{ k(k+1)}{2} (N \delta_N + \delta_N^{(P)}).
\end{align*}
We similarly obtain for Eve
\begin{align*}
\mathbb{P}[ O^{1:N}_{1:k} \neq \widehat{\widehat{O}}^{1:N}_{1:k}] 
& \leq \frac{ k(k+1)}{2}  (N \delta_N + \delta_N^{(P)}).
\end{align*}

Next, we show how Bob can recover the secret and private messages. Informally, the decoding process of the confidential and private messages $(M_{1:k},S_{1:k})$ for Bob is as follows. Reconstruction starts with Block $k$.  Given $(\Psi^{V|U}_{k}, \Phi^{V|U}_k, Y_{k}^{1:N},\widehat{U}^{1:N}_{k}) $, Bob can estimate $\widetilde{V}_k^{1:N}$, from which an estimate $\widehat{\Psi}^{V|U}_{k-1}$ of $\Psi^{V|U}_{k-1}$ is deduced. Then, for $i \in \llbracket 1 , k-1\rrbracket$, given $(\widehat{\Psi}^{V|U}_{k-i}, \Phi^{V|U}_{k-i}, Y_{k-i}^{1:N},\widehat{U}^{1:N}_{k-i}) $, Bob can estimate $\widetilde{V}_{k-i}^{1:N}$, from which an estimate of $\Psi^{V|U}_{k-i-1}$ is deduced. Finally, $S_{1:k}$ is formed from the estimate of $\widetilde{V}_{1:k}^{1:N}$. 

Formally, the analysis is as follows. For $i \in \llbracket 1 , k \rrbracket$, consider an optimal coupling~\cite[Lemma 3.6]{Aldous83} between $\widetilde{p}_{U_i^{1:N}V_i^{1:N}Y_i^{1:N}}$ and $p_{U^{1:N}V^{1:N}Y^{1:N}}$ such that $$\mathbb{P} [\mathcal{E}_{UVY,i}] = \mathbb{V}(\widetilde{p}_{U_i^{1:N}V_i^{1:N}Y_i^{1:N}} ,p_{U^{1:N}V^{1:N}Y^{1:N}}),$$ where $$\mathcal{E}_{UVY,i} \triangleq \{ (\widetilde{U}_i^{1:N}, \widetilde{V}_i^{1:N},{Y}_i^{1:N}) \neq ({U}^{1:N} , {V}^{1:N},{Y}^{1:N})\}.$$ Define also for $i \in \llbracket 1 , k-1 \rrbracket$, 
\begin{align*}
\mathcal{E}_{\Psi^{V|U}_i} &\triangleq \{ \widehat{\Psi}^{V|U}_i \neq \Psi^{V|U}_i\},\\
\mathcal{E}_{\widetilde{U}_i} & \triangleq \{ \widehat{U}_i^{1:N} \neq \widetilde{U}_i^{1:N}\},\\
\mathcal{E}_{i} &\triangleq \mathcal{E}_{\Psi^{V|U}_i} \cup  \mathcal{E}_{\widetilde{U}_i}. 
\end{align*}
For $i \in \llbracket 1 , k-1 \rrbracket$, we have
\begin{align*}
& \mathbb{P} [(M_{i}, S_i) \neq (\widehat{M}_{i},\widehat{S}_{i})]   \\
& \stackrel{(a)}{\leq} \mathbb{P} [\widetilde{V}_{i} \neq \widehat{V}_{i}]  \\
& = \mathbb{P} [\widetilde{V}_{i} \neq \widehat{V}_{i} | \mathcal{E}_{UVY,i}^c \cap \mathcal{E}_{{i}}^c ] \mathbb{P}[\mathcal{E}_{UVY,i}^c \cap \mathcal{E}_{{i}}^c] \\
& \phantom{mm }+  \mathbb{P} [\widetilde{V}_{i} \neq \widehat{V}_{i} |\mathcal{E}_{UVY,i} \cup \mathcal{E}_{{i}} ] \mathbb{P}[\mathcal{E}_{UVY,i} \cup \mathcal{E}_{{i}}]  \\
& \leq \mathbb{P} [\widetilde{V}_{i} \neq \widehat{V}_{i} |\mathcal{E}_{UVY,i}^c \cap \mathcal{E}_{{i}}^c ]  +   \mathbb{P}[\mathcal{E}_{UVY,i} \cup \mathcal{E}_{{i}}]  \\
& \leq \mathbb{P} [\widetilde{V}_{i} \neq \widehat{V}_{i} |\mathcal{E}_{UVY,i}^c \cap \mathcal{E}_{{i}}^c]  +   \mathbb{P}[\mathcal{E}_{UVY,i} ] + \mathbb{P}[ \mathcal{E}_{\Psi^{V|U}_i}] +  \mathbb{P}[ \mathcal{E}_{\widetilde{U}_i}]   \\
& \stackrel{(b)}{\leq} \mathbb{P} [\widetilde{V}_{i} \neq \widehat{V}_{i} |\mathcal{E}_{UVY,i}^c \cap \mathcal{E}_{{i}}^c]  +   \mathbb{P}[\mathcal{E}_{UVY,i} ] + \mathbb{P} [\widetilde{V}_{i+1} \neq \widehat{V}_{i+1}] \\
& \phantom{mm }+ \mathbb{P}[ \widehat{U}^{1:N}_{i} \neq \widetilde{U}^{1:N}_{i}]   \\
& \stackrel{(c)}{\leq} N \delta_N  +   \mathbb{P}[\mathcal{E}_{UVY,i} ] + \mathbb{P} [\widetilde{V}_{i+1} \neq \widehat{V}_{i+1}] + \mathbb{P}[ \widehat{U}^{1:N}_{i} \neq \widetilde{U}^{1:N}_{i}]   \\
& \stackrel{(d)}{\leq} N \delta_N  + \delta_N^{(P)} + \mathbb{P} [\widetilde{V}_{i+1} \neq \widehat{V}_{i+1}] +\mathbb{P}[ \widehat{U}^{1:N}_{i} \neq \widetilde{U}^{1:N}_{i}]  \\
& \stackrel{(e)}{\leq}  (i+1) \left( N \delta_N  + \delta_N^{(P)} \right) + \mathbb{P} [\widetilde{V}_{i+1} \neq \widehat{V}_{i+1}]  \\
& \stackrel{(f)}{\leq} (i+1) (k -i ) \left( N \delta_N  +  \delta_N^{(P)} \right) + \mathbb{P} [\widetilde{V}_{k} \neq \widehat{V}_{k}]  \\
& \stackrel{(g)}{\leq} (i+1)(k -i +1 ) \left( N \delta_N  + \delta_N^{(P)} \right) 
\end{align*}
where $(a)$ holds because $\widetilde{V}_{i}$ contains $(M_{i}, S_i)$ by construction, $(b)$ holds because $\widetilde{V}_{i+1}$ contains $\Psi^{V|U}_{i}$ by construction, $(c)$ follows from the error probability of lossless source coding with side information~\cite{Arikan10}, $(d)$ holds by the optimal coupling and Lemma~\ref{lem_dist_A}, $(e)$ holds by (\ref{eq_err_utilde}), $(f)$ holds by induction, $(g)$ is obtained similarly to the previous inequalities.

Hence,
\begin{align}
& \mathbb{P} [(M_{1:k},S_{1:k}) \neq (\widehat{M}_{1:k},\widehat{S}_{1:k})] \nonumber \\ \nonumber
& \leq \sum_{i=1}^k \mathbb{P} [(M_{i}, S_i) \neq (\widehat{M}_{i},\widehat{S}_{i})] \\ \nonumber
& \leq \sum_{i=1}^k (i+1) (k -i +1 ) \left( N \delta_N  +  \delta_N^{(P)} \right)\\
& = \left(\frac{k (k+1) (k+5)}{6} + k \right) \left( N \delta_N  + \delta_N^{(P)} \right).%
\end{align}

\subsection{Information leakage}

A Bayesian graph that describes dependencies between all the variables involved in the coding scheme of Section \ref{Sec_CS} is given in Figure~\ref{figFGD2}.
\begin{figure}
\centering
  \includegraphics[width=8.7cm]{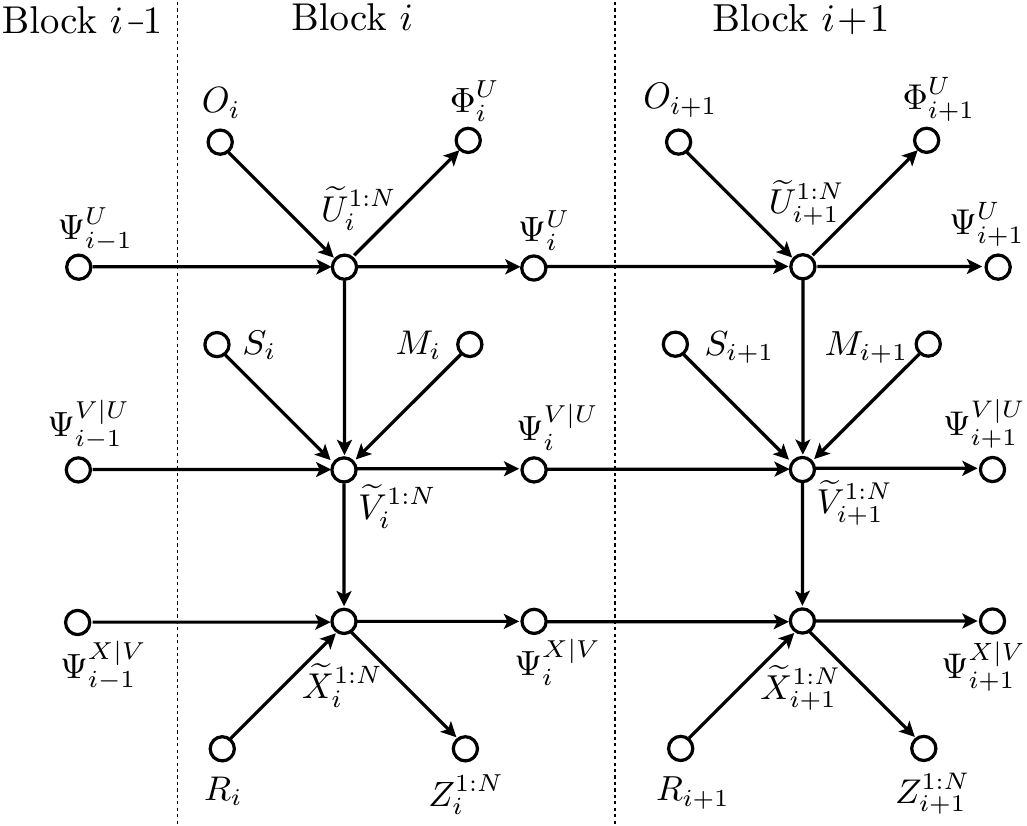}
  \caption{Graphical representation of the dependencies between consecutive encoding blocks. For Block $i \in \llbracket 1 ,k \rrbracket$, $O_i$ is the common message, $M_i$ is the private message, $S_i$ is the confidential message. $\Psi^{V|U}_i$ is the  information retransmitted in the next block to allow Bob to reconstruct $M_i$ and $S_i$ given $\Phi^{V|U}_i$ and its observations $Y^{1:N}_{1:k}$. $\Psi^{U}_i$ is the randomness used to form $\widetilde{U}_i^{1:N}$, $\Psi^{U}_i \subseteq \Psi^{U}_1$ is reused from the previous block. ${R}_i$ and $\Psi^{X|V}_i$ represent the randomness necessary at the encoder to form $\widetilde{X}_i^{1:N}$ where $\Psi^{X|V}_i = \Psi^{X|V}_1$ is reused from the previous block. Finally, $\Phi^{U}_i$ is information, whose rate is negligible, sent to Bob and Eve to allow them to reconstruct the common messages. }
  \label{figFGD2}
\end{figure}
For the secrecy analysis, we must upper bound
$$
 I(S_{1:k};\Psi^{U}_1\Phi^{U}_{1:k}Z_{1:k}^N).
$$
Note that we have introduced $(\Psi^{U}_1,\Phi^{U}_{1:k})$, since these random variables have been made available to Eve. Recall that $\Phi^{U}_{1:k}$ is additional information transmitted to Bob and Eve to reconstruct the common messages $O_{1:k}$. Recall also that $\Psi^{U}_1 \supset \Psi^{U}_i$, $i \in  \llbracket 2 , k \rrbracket$, as it is the randomness reused among all the blocks that allows the transmission of the common messages $O_{1:k}$.
We start by proving that secrecy holds for a given block $i \in \llbracket 2 , k \rrbracket$ in the following lemma.
\begin{lem} \label{lem1c}
For $i \in \llbracket 1 , k \rrbracket$ and $N$ large enough,
$$I(S_i \Psi^{V|U}_{i-1} ; Z_i^{1:N} \Phi^{U}_i \Psi^{U}_1) \leq \delta_N^{(*)},$$

where $\delta_N^{(*)} \triangleq  \sqrt{2\log 2} \sqrt{N \delta_N}(1+ 6 \sqrt{ 2 } + 3 \sqrt{3}) ( N - \log_2 (\sqrt{2\log 2} \sqrt{N \delta_N}(1+ 6 \sqrt{ 2 } + 3 \sqrt{3})) )$, and $\Psi^{V|U}_{0} \triangleq\emptyset$.
\end{lem}

\begin{proof}
See Appendix~\ref{App_lem1c}.
\end{proof}

Recall that for channel prefixing in the encoding process, we reuse some randomness $\Psi^{X|V}_1$ among all the blocks so that $\Psi^{X|V}_1 = \Psi^{X|V}_i$, $i \in  \llbracket 2 , k \rrbracket$. We show in the following lemma that $\Psi^{X|V}_1$ is almost independent from $(Z_{i}^{1:N}, \Psi^{V|U}_{i-1}, S_{i}, \Phi^{U}_i, \Psi^{U}_i)$. This fact will be useful in the secrecy analysis of the overall scheme.

\begin{lem} \label{lem4} 
For $i \in \llbracket 2 , k \rrbracket$ and $N$ large enough,
 $$I( \Psi^{X|V}_1 ; Z_{i}^{1:N} \Psi^{V|U}_{i-1} S_{i} \Phi^{U}_i \Psi^{U}_i ) \leq \delta_N^{(*)},$$
where $\delta_N^{(*)}$ is defined as in Lemma \ref{lem1c}.
\end{lem}
\begin{proof}
See Appendix~\ref{App_lem4}.
\end{proof}

Using Lemmas \ref{lem1c} and \ref{lem4}, we show in the following lemma a recurrence relation that will make the secrecy analysis over all blocks easier.

\begin{lem} \label{lemdifc}
Let $i \in \llbracket 1 , k-1 \rrbracket$. Define $$\widetilde{L}_{i} \triangleq I(S_{1:k}; \Psi^{U}_1 \Phi^{U}_{1:i} Z_{1:i}^{1:N}).$$ We have
$$
\widetilde{L}_{i+1}  - \widetilde{L}_{i} \leq 3\delta_N^{(*)}.
$$
\end{lem}

\begin{proof}
See Appendix~\ref{App_lemdifc}.
\end{proof}

We then have
\begin{align*}
\widetilde{L}_{1}
& = I(S_{1:k}; \Psi^{U}_1 \Phi^{U}_{1} Z_{1}^{1:N}) \\
& = I(S_1 ; \Psi^{U}_1 \Phi^{U}_{1} Z_1^{1:N})  +  I(S_{2:k}; \Psi^{U}_1 \Phi^{U}_{1} Z_{1}^{1:N} |S_1)\\
& \stackrel{(a)}{\leq} \delta_N^{(*)} +I(S_{2:k}; \Psi^{U}_1 \Phi^{U}_{1} Z_{1}^{1:N} |S_1)\\
& \leq  \delta_N^{(*)} +I(S_{2:k};  \Psi^{U}_1 \Phi^{U}_{1} Z_{1}^{1:N}S_1 )\\
& \stackrel{(b)}{=} \delta_N^{(*)},
\end{align*}
where $(a)$ follows from Lemma \ref{lem1c}, $(b)$ follows from independence of $S_{2:k}$ and the random variables of Block 1.

Hence, strong secrecy follows from Lemma \ref{lemdifc} because
\begin{align*}
 I(S_{1:k};\Psi^{U}_1\Phi^{U}_{1:k}Z_{1:k}^{1:N}) 
& = \widetilde{L}_{1} + \sum_{i=1}^{k-1} (\widetilde{L}_{i+1}  - \widetilde{L}_{i}) \\
& \leq \delta_N^{(*)} + (k-1) (3\delta_N^{(*)}) \\
& = (3k-2) \delta_N^{(*)}.
\end{align*}

\section{Conclusion}
\label{sec:conclusion}
Our proposed polar coding scheme for the broadcast channel with confidential messages provides an explicit low-complexity scheme achieving the capacity region of~\cite{Watanabe12}, and uses the optimal amount of local randomness at the stochastic encoder. Although the presence of auxiliary random variables and the need to re-align polarization sets through chaining introduces rather involved notation, the coding scheme is conceptually close to a binning proof of the capacity region, in which polarization is used in place of random binning. We believe that a systematic use of this connection will effectively allow one to translate many results proved with output statistics of random binning\cite{yassaee2014achievability} into polar coding schemes.

It is arguable whether the resulting schemes are truly practical, as the block length $N$ and the number of blocks $k$ are likely to be fairly large. Although only random seeds with negligible rate need to be shared between the transmitter and receivers, much work remains to be done to circumvent the need for such seeds.

\appendices
\section{Proof of Lemma~\ref{lem_dist_A}} \label{App_lem_dist_A}
In the following, for joint probability distributions $p_{XY}$ and $q_{XY}$ defined over $\mathcal{X}\times \mathcal{Y}$, we write the conditional relative entropy as
\begin{multline*}
\mathbb{E}_{p_X} \left[ \mathbb{D}(p_{Y|X}||q_{Y|X}) \right]  \triangleq \sum_{x\in \mathcal{X}} p_X(x)\mathbb{D}(p_{Y|X=x}||q_{Y|X=x}).
\end{multline*}

We show the first three inequalities of Lemma \ref{lem_dist_A} in order. Let $i \in \llbracket 2,k-1\rrbracket$.
We have
\begin{align}
& \mathbb{D}(p_{U^{1:N}} || \widetilde{p}_{U_i^{1:N}})  \nonumber  \\ \nonumber
& \stackrel{(a)}{=} \mathbb{D}(p_{A^{1:N}} || \widetilde{p}_{A_i^{1:N}})  \nonumber \displaybreak[0]\\ 
& \stackrel{(b)}{=}  \sum_{j=1}^N \mathbb{E}_{p_{A^{1:j-1}}} \left[ \mathbb{D}(p_{A^{j}|A^{1:j-1}} || \widetilde{p}_{A_i^j|A_i^{1:j-1}} )\right] \nonumber\displaybreak[0] \\ 
& \stackrel{(c)}{=}  \sum_{j\in \mathcal{V}_{U}}  \mathbb{E}_{p_{A^{1:j-1}}}\left[ \mathbb{D}(p_{A^{j}|A^{1:j-1}} || \widetilde{p}_{A_i^j|A_i^{1:j-1}} )\right] \nonumber \\
& \stackrel{(d)}{=} \sum_{j \in \mathcal{V}_{U}} ( \log_2(q^{(U)}) -H(A^{j}|A^{1:j-1}) )  \nonumber\\ \nonumber
& \stackrel{(e)}{\leq} |\mathcal{V}_{U}| \delta_N  \\
& \leq N \delta_N, \label{eq_distA}
\end{align}
where $(a)$ holds by invertibility of $G_n$, $(b)$ holds by the chain rule for divergence \cite{Cover91}, $(c)$ holds by (\ref{eq_sim_A_i}), $(d)$ holds by (\ref{eq_sim_A_i}) and uniformity of $O_i$, $O_{i-1,2}$, and $\Psi_1^U$, $(e)$ holds by definition of $\mathcal{V}_{U}$.

Similarly for $i \in \{ 1, k \}$, using (\ref{eq_sim_A_1}) and (\ref{eq_sim_A_k}) we also have
\begin{align}
 \mathbb{D}(p_{U^{1:N}} || \widetilde{p}_{U_i^{1:N}}) 
 \leq N \delta_N. \label{eq_distA2}
\end{align}

Let $i \in \llbracket 2,k \rrbracket$. We have
\begin{align}
&  \mathbb{E}_{p_{U^{1:N}}} \left[ \mathbb{D}(p_{B^{1:N}|U^{1:N}} || \widetilde{p}_{B_i^{1:N}|U_i^{1:N}}) \right] \nonumber \\ \nonumber
& \stackrel{(a)}{=} \sum_{j=1}^N \mathbb{E}_{p_{B^{1:j-1}U^{1:N}}}\left[ \mathbb{D}(p_{B^{j}|B^{1:j-1}U^{1:N}} || \widetilde{p}_{B_i^j|B_i^{1:j-1}U_i^{1:N}}) \right] \\ \nonumber
& \stackrel{(b)}{=} \sum_{j\in \mathcal{V}_{V|U}} \!\!\!\mathbb{E}_{p_{B^{1:j-1}U^{1:N}}}\left[ \mathbb{D}(p_{B^{j}|B^{1:j-1}U^{1:N}} || \widetilde{p}_{B_i^j|B_i^{1:j-1}U_i^{1:N}}) \right]\\ \nonumber
& \stackrel{(c)}{=} \sum_{j\in \mathcal{V}_{V|U}} ( \log_2(q^{(V)}) -H(B^{j}|B^{1:j-1}U^{1:N}) ) \\ \nonumber
& \stackrel{(d)}{\leq} |\mathcal{V}_{V|U}| \delta_N\\
&  \leq N \delta_N, \label{eqsupbv1}
\end{align}
where $(a)$ holds by the chain rule, $(b)$ holds by~(\ref{defsimBi}), $(c)$ holds by~(\ref{defsimBi}) and uniformity of $\Psi^{V|U}_{i-1}$, $S_i$, and $M_i$, $(d)$ holds by definition of $\mathcal{V}_{V|U}$.

 Then,
 \begin{align}
& \mathbb{D}(p_{V^{1:N}U^{1:N}} || \widetilde{p}_{V_i^{1:N}U_i^{1:N}})   \nonumber \\  \nonumber
&  \stackrel{(a)}{=} \mathbb{D}(p_{B^{1:N}U^{1:N}} || \widetilde{p}_{B_i^{1:N}U_i^{1:N}})  \\  \nonumber
& \stackrel{(b)}{=}  \mathbb{E}_{p_{U^{1:N}}} \left[ \mathbb{D}(p_{B^{1:N}|U^{1:N}} || \widetilde{p}_{B_i^{1:N}|U_i^{1:N}})\right] + \mathbb{D}(p_{U^{1:N}}|| \widetilde{p}_{U_i^{1:N}}) \\
& \stackrel{(c)}{\leq} 2 N \delta_N, \label{eq_distB}
 \end{align}
 where $(a)$ holds by invertibility of $G_n$, $(b)$ holds by the chain rule, $(c)$ holds by \eqref{eq_distA}, \eqref{eq_distA2}, and (\ref{eqsupbv1}).

Similarly, using \eqref{eq_distA2}, and (\ref{eq_sim_Bv_1}), we have
\begin{align}
 \mathbb{D}(p_{V^{1:N}U^{1:N}}|| \widetilde{p}_{V_1^{1:N}U_1^{1:N}}) 
{\leq} 2 N \delta_N. \label{eq_distB2}
 \end{align}

Let $i \in \llbracket 2,k \rrbracket$. We have
\begin{align}
&  \mathbb{E}_{p_{V^{1:N}}} \left[ \mathbb{D}(p_{T^{1:N}|V^{1:N}} || \widetilde{p}_{T_i^{1:N}|V_i^{1:N}}) \right]\nonumber \displaybreak[0] \\ \nonumber
& \stackrel{(a)}{=} \sum_{j=1}^N \mathbb{E}_{p_{T^{1:j-1}V^{1:N}}}\left[\mathbb{D}(p_{T^{j}|T^{1:j-1}V^{1:N}}||  \widetilde{p}_{T_i^j|T_i^{1:j-1}V_i^{1:N}})\right] \displaybreak[0] \\ \nonumber
& \stackrel{(b)}{=} \sum_{j\in \mathcal{V}_{X|V}} \!\!\!\mathbb{E}_{p_{T^{1:j-1}V^{1:N}}}\left[\mathbb{D}(p_{T^{j}|T^{1:j-1}V^{1:N}} || \widetilde{p}_{T_i^j|T_i^{1:j-1}V_i^{1:N}})\right] \displaybreak[0]
\\ \nonumber
& \stackrel{(c)}{=} \sum_{j\in \mathcal{V}_{X|V}} ( \log_2(q^{(X)}) -H(T^{j}|T^{1:j-1}V^{1:N}) ) \\ \nonumber
& \stackrel{(d)}{\leq} |\mathcal{V}_{X|V}| \delta_N\\
&  \leq N \delta_N, \label{eqsupd1XV}
\end{align}
where $(a)$ holds by the chain rule, $(b)$ holds by~(\ref{defsimTi}), $(c)$ holds by~(\ref{defsimTi}) and uniformity of the $q^{(X)}$-ary symbols in $\widetilde{T}_i^{1:N}[ \mathcal{V}_{X|V}]$, $(d)$ holds by definition of $\mathcal{V}_{X|V}$.
 
 Then,
 \begin{align}
& \mathbb{D}(p_{X^{1:N}V^{1:N}} || \widetilde{p}_{X_i^{1:N}V_i^{1:N}})  \nonumber \\ \nonumber
&  \stackrel{(a)}{=}  \mathbb{D}(p_{T^{1:N}V^{1:N}} || \widetilde{p}_{T_i^{1:N}V_i^{1:N}})  \\ \nonumber
& \stackrel{(b)}{=} \mathbb{E}_{p_{V^{1:N}}} \left[ \mathbb{D}(p_{T^{1:N}|V^{1:N}} || \widetilde{p}_{T_i^{1:N}|V_i^{1:N}}) \right]+ \mathbb{D}(p_{V^{1:N}} || \widetilde{p}_{V_i^{1:N}}) \\
& \stackrel{(c)}{\leq} 3 N \delta_N, \label{eq_distT}
 \end{align}
 where $(a)$ holds by invertibility of $G_n$, $(b)$ holds by the chain rule, $(c)$ holds by \eqref{eq_distB} and (\ref{eqsupd1XV}) .
 
 Similarly, using (\ref{defsimT_1}) and~\eqref{eq_distB2}, we have
\begin{align}
 \mathbb{D}(p_{X^{1:N}V^{1:N}} || \widetilde{p}_{X_1^{1:N}V_1^{1:N}}) 
{\leq} 3 N \delta_N. \label{eq_distT2}
 \end{align}

Note that, as remarked in \cite{Goela13}, upper-bounding the divergence with a chain rule is easier than directly upper-bounding the variational distance as in \cite{Korada10,Honda13}.

Using \eqref{eq_distA}, \eqref{eq_distA2}, \eqref{eq_distB}, \eqref{eq_distB2}, \eqref{eq_distT}, \eqref{eq_distT2}, we now prove the last inequality in Lemma \ref{lem_dist_A}. Let $i \in \llbracket 1,k \rrbracket$. Because of the Markov chains $U \to V \to X \to (YZ)$ and $\widetilde{U}_i^{1:N} \to \widetilde{V}_i^{1:N} \to \widetilde{X}_i^{1:N} \to (Y_i^{1:N}Z_i^{1:N})$, we have
\begin{align*}
&	p_{U^{1:N}V^{1:N}X^{1:N}Y^{1:N}Z^{1:N}} \\ 
& \phantom{mmmmm}= p_{Y^{1:N}Z^{1:N}|X^{1:N}}p_{X^{1:N}|V^{1:N}} p_{U^{1:N}V^{1:N}},\\
& \widetilde{p}_{U_i^{1:N}V_i^{1:N}X_i^{1:N}Y_i^{1:N}Z_i^{1:N}} \\ 
& \phantom{mmmmm}= \widetilde{p}_{Y_i^{1:N}Z_i^{1:N}|X_i^{1:N}} \widetilde{p}_{X_i^{1:N}|V_i^{1:N}} \widetilde{p}_{U_i^{1:N}V_i^{1:N}}.
\end{align*}

Hence, since $p_{Y^{1:N}Z^{1:N}|X^{1:N}} =\widetilde{p}_{Y_i^{1:N}Z_i^{1:N}|X_i^{1:N}}$, we have by \cite[Lemma 17]{Cuff09}
\begin{align}
& \mathbb{V}(p_{U^{1:N}V^{1:N}X^{1:N}Y^{1:N}Z^{1:N}}, \widetilde{p}_{U_i^{1:N}V_i^{1:N}X_i^{1:N}Y_i^{1:N}Z_i^{1:N}}) \nonumber  \\    
& = \mathbb{V}(p_{X^{1:N}|V^{1:N}} p_{U^{1:N}V^{1:N}}, \widetilde{p}_{X_i^{1:N}|V_i^{1:N}} \widetilde{p}_{U_i^{1:N}V_i^{1:N}} ).\label{eq:part1}
 \end{align}

We also have
\begin{align}
& \mathbb{V}(p_{X^{1:N}|V^{1:N}} p_{U^{1:N}V^{1:N}}, \widetilde{p}_{X_i^{1:N}|V_i^{1:N}} p_{U^{1:N}V^{1:N}} ) \nonumber \\ \nonumber
& = \mathbb{V}(p_{X^{1:N}|V^{1:N}} p_{V^{1:N}}, \widetilde{p}_{X_i^{1:N}|V_i^{1:N}} p_{V^{1:N}} ) \\ \nonumber
& \stackrel{(a)}{\leq} \mathbb{V}(p_{X^{1:N}|V^{1:N}} p_{V^{1:N}}, \widetilde{p}_{X_i^{1:N}V_i^{1:N}}  ) \\ \nonumber
& \phantom{mm}+ \mathbb{V}(\widetilde{p}_{X_i^{1:N}V_i^{1:N}}, \widetilde{p}_{X_i^{1:N}|V_i^{1:N}} p_{V^{1:N}} )\\ \nonumber
& = \mathbb{V}(p_{X^{1:N}V^{1:N}} , \widetilde{p}_{X_i^{1:N}V_i^{1:N}}  ) + \mathbb{V}(\widetilde{p}_{V_i^{1:N}},  p_{V^{1:N}} )\\ \nonumber
& \leq  \mathbb{V}(p_{X^{1:N}V^{1:N}} , \widetilde{p}_{X_i^{1:N}V_i^{1:N}}  ) + \mathbb{V}( p_{U^{1:N}V^{1:N}},  \widetilde{p}_{U_i^{1:N}V_i^{1:N}} )\\ 
& \stackrel{(b)}{\leq}   \delta_N^{(XV)} + \delta_N^{(UV)},  \label{eq:part2}
\end{align}
where $(a)$ holds by the triangle inequality, and $(b)$ holds by~\eqref{eq_distB}, \eqref{eq_distB2} and \eqref{eq_distT}, \eqref{eq_distT2} using Pinsker's inequality.

Finally, we have
\begin{align}
& \mathbb{V}(p_{U^{1:N}V^{1:N}X^{1:N}Y^{1:N}Z^{1:N}}, \widetilde{p}_{U_i^{1:N}V_i^{1:N}X_i^{1:N}Y_i^{1:N}Z_i^{1:N}}) \nonumber  \\  \nonumber
& \stackrel{(a)}{\leq} \mathbb{V}(p_{X^{1:N}|V^{1:N}} p_{U^{1:N}V^{1:N}}, \widetilde{p}_{X_i^{1:N}|V_i^{1:N}} p_{U^{1:N}V^{1:N}} ) \\ \nonumber
& \phantom{mm} + \mathbb{V}(\widetilde{p}_{X_i^{1:N}|V_i^{1:N}} p_{U^{1:N}V^{1:N}}, \widetilde{p}_{X_i^{1:N}|V_i^{1:N}} \widetilde{p}_{U_i^{1:N}V_i^{1:N}} )  \\ \nonumber
& = \mathbb{V}(p_{X^{1:N}|V^{1:N}} p_{U^{1:N}V^{1:N}}, \widetilde{p}_{X_i^{1:N}|V_i^{1:N}} p_{U^{1:N}V^{1:N}} ) \\ \nonumber
& \phantom{mm}+ \mathbb{V}( p_{U^{1:N}V^{1:N}},  \widetilde{p}_{U_i^{1:N}V_i^{1:N}} )  \\ \nonumber
& \stackrel{(b)}{\leq} \delta_N^{(XV)} + 2 \delta_N^{(UV)},
\end{align}
where $(a)$ holds by the triangle inequality and Equation \eqref{eq:part1}, $(b)$ holds by \eqref{eq_distB}, \eqref{eq_distB2}, and \eqref{eq:part2} using Pinsker's inequality. %

\section{Randomization in (\ref{eq_sim_A_1})--(\ref{defsimTi})} \label{App_lemrandA}
We here justify that the rate of uniform randomness required for successive cancellation encoding in~(\ref{eq_sim_A_1})--(\ref{defsimTi}) is negligible. We will make use of the following lemma.
\begin{lem} \label{lemcoin}
	Let $N \in \mathbb{N}$, and let $\mathcal{J}_N$ be a subset of $\llbracket 1,N \rrbracket$ such that
	\begin{align}
\lim_{N \to \infty}    \frac{|\mathcal{J}_N|}{N} =0. \label{eq:hyp1}
	\end{align}
 Consider $|\mathcal{J}_N|$ sources indexed by $j \in \mathcal{J}_N$, $(\mathcal{X}\times \mathcal{Y},p_{X_jY_j})$ where $\mathcal{X}$ and $\mathcal{Y}$ are finite alphabets. 

	Let $p_U$ denote the uniform distribution over $\mathcal{X}$. We call a sample drawn from $p_U$ a coin toss. Using the interval algorithm~\cite{hoshi1997interval} and assuming that for $j \in \mathcal{J}_N$,  $y_j$ is drawn from $\widetilde{p}_{Y_j}$, one can sample from $p_{X_j|Y_j=y_j} $ using $L_j$ independent coin tosses such that for any $\epsilon>0$ with probability arbitrarily close to one as $N$ goes to infinity,
	$$
\frac{ \sum_{j \in \mathcal{J}_N} \mathbb{E}_{\widetilde{p}_{Y_j}} [L_{j}]}{N} < \epsilon.
	$$
\end{lem}
 
\begin{proof}
For any $j \in \mathcal{J}_N$, using the interval algorithm by~\cite[Theorem 3]{hoshi1997interval}, one can sample from $p_{X_j|Y_j=y_j}$ using $L_{j}$ independent coin tosses with an expected number of coin tosses upper-bounded as follows.
\begin{align}
\mathbb{E}[L_{j}] \leq \frac{H(X_j|Y_j=y_j)}{\log |\mathcal{X}|} + \frac{|\mathcal{X}|}{|\mathcal{X}|-1} + \frac{\log 2}{\log |\mathcal{X}|}. \label{eq:hoshi1}
\end{align}
From \eqref{eq:hoshi1}, we obtain the trivial upper bound
\begin{align*}
\mathbb{E}[L_{j}] \leq 1 + \frac{|\mathcal{X}|}{|\mathcal{X}|-1} + \frac{\log 2}{\log |\mathcal{X}|}. 
\end{align*}
We thus have
\begin{align}
 \mathbb{E}\left[ \frac{\sum_{j \in \mathcal{J}_N} \mathbb{E}_{\widetilde{p}_{Y_j}} [L_{j}] }{N}\right]\nonumber\nonumber
& 
= \frac{1}{N} \sum_{j \in \mathcal{J}_N} \mathbb{E}_{\widetilde{p}_{Y_j}}  [ \mathbb{E}\left[ L_{j}] \right]\nonumber \\\nonumber%
& \leq \frac{|\mathcal{J}_N|}{N}  \left[ 1 + \frac{|\mathcal{X}|}{|\mathcal{X}|-1} +\frac{\log 2}{\log |\mathcal{X}|}  \right] \\  
& \xrightarrow{N \to \infty}{0}, \label{eqlim}
\end{align}
and we conclude with Markov's inequality.
\end{proof}

We start by studying the rate of uniform randomness required for successive cancellation encoding in~(\ref{eq_sim_A_1}), (\ref{eq_sim_A_i}), and~(\ref{eq_sim_A_k}). For any $i\in \llbracket 1,k\rrbracket$, note that the random decisions in~(\ref{eq_sim_A_1}), (\ref{eq_sim_A_i}), and~(\ref{eq_sim_A_k}),  
$$
\widetilde{p}_{{A}_i^j|{A}_i^{1:j-1}} ({a}_i^j|{a}_i^{1:j-1}) \triangleq {p}_{A^j|A^{1:j-1}} (a_i^j|a_i^{1:j-1})  \text{if }j\in {\mathcal{V}}_{U}^c,
$$
can be replaced, using the result in \cite{Chou15f}, by
\begin{multline*}
	\widetilde{p}_{{A}_i^j|{A}_i^{1:j-1}} ({a}_i^j|{a}_i^{1:j-1}) \\ \triangleq \begin{cases}
	{p}_{A^j|A^{1:j-1}} (a_i^j|a_i^{1:j-1})   & \text{if }j\in  {\mathcal{V}}^c_{U} \backslash {\mathcal{H}}_{U}^c\\
	\mathds{1} \left\{ {a}_i^j = (a_i^j)^* \right\}& \text{if }j\in   {\mathcal{H}}_{U}^c
	\end{cases},
\end{multline*}
where $(a_i^j)^* \triangleq \displaystyle\argmax_{a}{p}_{A^j|A^{1:j-1}} (a|a_i^{1:j-1}) $.
Hence, the rate of uniform randomness required for successive cancellation encoding in~(\ref{eq_sim_A_1}), (\ref{eq_sim_A_i}), and~(\ref{eq_sim_A_k}) is negligible, with probability arbitrarily close to one, by Lemma \ref{lemcoin} applied with the substitutions $\mathcal{J_N} \leftarrow\mathcal{H}_{U} \backslash {\mathcal{V}}_{U}$, $X_j \leftarrow A^j$, $Y_j \leftarrow A^{1:j-1}$, where $j \in {\mathcal{V}}^c_{U} \backslash {\mathcal{H}}_{U}^c$. The assumption of Lemma \ref{lemcoin} is indeed satisfied since by Lemma \ref{lemcard_1} and Lemma \ref{lemcard_2}, 
\begin{align*}
\frac{|{\mathcal{V}}^c_{U} \backslash {\mathcal{H}}_{U}^c |}{N} 
& = \frac{|{\mathcal{V}}^c_{U}|}{N} -\frac{| {\mathcal{H}}^c_{U} | }{N} \\
&\xrightarrow{N \to \infty}{0}.
\end{align*}

Similarly, for any $i\in \llbracket 1,k\rrbracket$, the random decisions in (\ref{eq_sim_Bv_1}) and (\ref{defsimBi}),  
\begin{multline*}
\widetilde{p}_{{B}_i^j|{B}_i^{1:j-1} U_i^{1:N}} ({b}_i^j|{b}_i^{1:j-1}\widetilde{u}_{i}^{1:N}) \\ \triangleq
  {p}_{B^j|B^{1:j-1}U^{1:N}} (b_i^j|b_i^{1:j-1}\widetilde{u}_{i}^{1:N})  \text{ if }j\in {\mathcal{V}}_{V|U}^c,
\end{multline*}
can be replaced, using the result in \cite{Chou15f}, by
\begin{align*}
&\widetilde{p}_{{B}_i^j|{B}_i^{1:j-1} U_i^{1:N}} ({b}_i^j|{b}_i^{1:j-1}\widetilde{u}_{i}^{1:N})  \\ 
& \triangleq  \begin{cases}
	{p}_{B^j|B^{1:j-1}U^{1:N}} (b_i^j|b_i^{1:j-1}\widetilde{u}_{i}^{1:N})   &\text{if }j\in  {\mathcal{V}}_{V|U}^c \backslash {\mathcal{H}}^c_{V|U}\\
	\mathds{1} \left\{ {b}_i^j = ({b}_i^j)^* \right\}  &\text{if }j\in   {\mathcal{H}}_{V|U}^c
	\end{cases},
\end{align*}
where $({b}_i^j)^*\triangleq \displaystyle\argmax_{b}{p}_{B^j|B^{1:j-1}U^{1:N}} (b|b_i^{1:j-1}\widetilde{u}_{i}^{1:N})$,
and
for any $i\in \llbracket 1,k\rrbracket$, the random decisions in (\ref{defsimT_1}) and (\ref{defsimTi}),  
\begin{multline*}
\widetilde{p}_{T^j_i|T_i^{1:j-1}V_i^{1:N}} (t_i^j|t_i^{1:j-1}\widetilde{v}_i^{1:N}) \\ \triangleq
  {p}_{T^j|T^{1:j-1}V^{1:N}} (t_i^j|t_i^{1:j-1}\widetilde{v}_i^{1:N})  \text{ if }j \in \mathcal{V}_{X|V}^c,
\end{multline*}
can be replaced, using the result in \cite{Chou15f}, by
\begin{align*}
	&\widetilde{p}_{T^j_i|T_i^{1:j-1}V_i^{1:N}} (t_i^j|t_i^{1:j-1}\widetilde{v}_i^{1:N}) \\
	& \triangleq \begin{cases}
	{p}_{T^j|T^{1:j-1}V^{1:N}} (t_i^j|t_i^{1:j-1}\widetilde{v}_i^{1:N})  & \text{if }j\in  {\mathcal{V}}^c_{X|V} \backslash {\mathcal{H}}^c_{X|V}\\
	\mathds{1} \left\{ {t}_i^j = ({t}_i^j)^* \right\} & \text{if }j\in   {\mathcal{H}}_{X|V}^c
	\end{cases},
\end{align*}
where $ ({t}_i^j)^* \triangleq \displaystyle\argmax_{t}{p}_{T^j|T^{1:j-1}V^{1:N}} (t|t_i^{1:j-1}\widetilde{v}_i^{1:N})$.

Hence, the rate of uniform randomness required for successive cancellation encoding in~(\ref{eq_sim_Bv_1})--(\ref{defsimTi}) is negligible, with probability arbitrarily close to one, by Lemma \ref{lemcoin} applied with the substitutions $\mathcal{J_N} \leftarrow {\mathcal{V}}_{V|U}^c \backslash {\mathcal{H}}^c_{V|U}$, $X_j \leftarrow B^j$, $Y_j \leftarrow (B^{1:j-1},U^{1:N})$, where $j \in {\mathcal{V}}_{V|U}^c \backslash {\mathcal{H}}^c_{V|U}$, and by Lemma \ref{lemcoin} applied with the substitutions $\mathcal{J_N} \leftarrow {\mathcal{V}}^c_{X|V} \backslash {\mathcal{H}}^c_{X|V}$, $X_j \leftarrow T^j$, $Y_j \leftarrow (T^{1:j-1},V^{1:N})$, where $j \in {\mathcal{V}}^c_{X|V} \backslash {\mathcal{H}}^c_{X|V}$.
\begin{rem}
The question whether the randomized decisions for the bits in positions ${\mathcal{V}}^c_{U} \backslash {\mathcal{H}}_{U}^c$, ${\mathcal{V}}_{V|U}^c \backslash {\mathcal{H}}^c_{V|U}$, and ${\mathcal{V}}^c_{X|V} \backslash {\mathcal{H}}^c_{X|V}$, can be replaced by deterministic decisions, remains open \cite{Chou15f}. 
\end{rem}

\section{Proof of Lemma~\ref{lem1c}} \label{App_lem1c}

We will use of the following lemma.

\begin{lem} \label{lemdrev3}
Consider the random variables $(F,G)$ distributed according to $p_{FG}$ over the alphabets $\mathcal{F} \times \mathcal{G}$, where $|\mathcal{F} | = q^{(F)}$, with $q^{(F)}$ prime. Consider $N$ independent realizations of these random variables $F^{1:N}$ and $G^{1:N}$.
Consider the random variables $(\widetilde{F}^{1:N},\widetilde{G}^{1:N})$ distributed according to $\widetilde{p}_{F^{1:N}G^{1:N}}$ over the alphabets $ \mathcal{F}^N \times \mathcal{G}^N$. Define $\widetilde{E}^{1:N} \triangleq \widetilde{F}^{1:N} G_n$ and ${E}^{1:N} \triangleq {F}^{1:N} G_n$. Define also 
 $$\mathcal{V}_{F|G} \triangleq  \left\{ i \in \llbracket 1,N \rrbracket: H( E^i | E^{1:i-1} G^{1:N}) >  \log_2(q^{(F)}) - \delta_N  \right\},$$ with $\delta_N\eqdef 2^{-N^\beta}$ and $\beta\in]0,\tfrac{1}{2}[$.

Assume that 
$$\mathbb{V}(p_{E^{1:N}G^{1:N}}, \widetilde{p}_{E^{1:N}G^{1:N}}) \leq \delta_N^{(FG)}.$$

Then, we have 
\begin{multline*}
\mathbb{V}(\widetilde{p}_{E^{1:N}[\mathcal{V}_{F|G}] G^{1:N}},\widetilde{p}_{E^{1:N}[\mathcal{V}_{F|G}]} \widetilde{p}_{G^{1:N}})  \\ \leq  \sqrt{2\log 2} \sqrt{N \delta_N} + 3\delta_N^{(FG)}.
\end{multline*}
\end{lem}

\begin{proof}
	We have 
	\begin{align}
&  \mathbb{V}({p}_{E^{1:N}[\mathcal{V}_{F|G}] G^{1:N} },\widetilde{p}_{E^{1:N}[\mathcal{V}_{F|G}]} \widetilde{p}_{ G^{1:N}})   \nonumber\\  \nonumber
&  \stackrel{(a)}{\leq} \mathbb{V}({p}_{E^{1:N}[\mathcal{V}_{F|G}] G^{1:N}},{p}_{E^{1:N}[\mathcal{V}_{F|G}]} p_{ G^{1:N}}) \\ \nonumber
& \phantom{mm} + \mathbb{V}({p}_{E^{1:N}[\mathcal{V}_{F|G}]} p_{ G^{1:N}},\widetilde{p}_{E^{1:N}[\mathcal{V}_{F|G}]} \widetilde{p}_{ G^{1:N}}) \\  \nonumber
& \stackrel{(b)}{\leq} \mathbb{V}({p}_{E^{1:N}[\mathcal{V}_{F|G}] G^{1:N}},{p}_{E^{1:N}[\mathcal{V}_{F|G}]} p_{ G^{1:N}}) \\ \nonumber
& \phantom{mm}+ \mathbb{V}({p}_{E^{1:N}[\mathcal{V}_{F|G}]} ,\widetilde{p}_{E^{1:N}[\mathcal{V}_{F|G}]} ) +  \mathbb{V}(p_{ G^{1:N}},\widetilde{p}_{ G^{1:N}}) \\  \nonumber
& \stackrel{(c)}{\leq} \mathbb{V}({p}_{E^{1:N}[\mathcal{V}_{F|G}] G^{1:N}},{p}_{E^{1:N}[\mathcal{V}_{F|G}]} p_{ G^{1:N}})  + 2 \delta_N^{(FG)} \\  \nonumber
& \stackrel{(d)}{\leq} \sqrt{2 \log2} \sqrt{ \mathbb{D}({p}_{E^{1:N}[\mathcal{V}_{F|G}] G^{1:N}} || {p}_{E^{1:N}[\mathcal{V}_{F|G}]} p_{G^{1:N}})}  \\ \nonumber
& \phantom{mm} + 2 \delta_N^{(FG)} \\  \nonumber
& = \sqrt{2 \log2} \sqrt{ I( E^{1:N}[\mathcal{V}_{F|G}] ;G^{1:N} )}  + 2 \delta_N^{(FG)} \\  
& \stackrel{(e)}{\leq}   \sqrt{2\log 2} \sqrt{N \delta_N} + 2\delta_N^{(FG)}, \label{eq_sec_int2}
\end{align}
where $(a)$ and $(b)$ follow from the triangle inequality, $(c)$ holds by hypothesis, $(d)$ holds by Pinsker's inequality, $(e)$ holds because using the fact that conditioning reduces entropy we have
\begin{align*}
& I( E^{1:N}[\mathcal{V}_{F|G}] ; G^{1:N}  )\\
& =  H( E^{1:N}[\mathcal{V}_{F|G}] ) - H({E}^{1:N}[ \mathcal{V}_{F|G} ]| G^{1:N})\\
& \leq |\mathcal{V}_{F|G}| \log_2(q^{(F)})- \sum_{j \in \mathcal{V}_{F|G} } H(E^{j} | E^{1:j-1} G^{1:N})\\
& \leq |\mathcal{V}_{F|G}|\log_2(q^{(F)}) + |\mathcal{V}_{F|G}| (\delta_N -\log_2(q^{(F)}))\\
& \leq N \delta_N.
\end{align*}
We then obtain
\begin{align}
&\mathbb{V}(\widetilde{p}_{E^{1:N}[\mathcal{V}_{F|G}] G^{1:N}},\widetilde{p}_{E^{1:N}[\mathcal{V}_{F|G}]} \widetilde{p}_{G^{1:N}}) \nonumber \displaybreak[0]\\ \nonumber
& \stackrel{(a)}{\leq} \mathbb{V}(\widetilde{p}_{E^{1:N}[\mathcal{V}_{F|G}] G^{1:N}},{p}_{E^{1:N}[\mathcal{V}_{F|G}] G^{1:N}})\\ \nonumber
& \phantom{mm}
+  \mathbb{V}({p}_{E^{1:N}[\mathcal{V}_{F|G}] G^{1:N}},\widetilde{p}_{E^{1:N}[\mathcal{V}_{F|G}]} \widetilde{p}_{ G^{1:N}})  \displaybreak[0] \\
& \stackrel{(b)}{\leq}  \sqrt{2\log 2} \sqrt{N \delta_N} + 3\delta_N^{(FG)}, \label{eq_sec_int3}
\end{align}
where $(a)$ holds by the triangle inequality, $(b)$ holds by hypothesis, and (\ref{eq_sec_int2}).
\end{proof}
Let $i \in \llbracket 1 , k \rrbracket$. With the substitution $F^{1:N} \leftarrow V^{1:N}$, $E^{1:N} \leftarrow B^{1:N}$, $G^{1:N} \leftarrow (U^{1:N}Z^{1:N})$, $\widetilde{F}^{1:N} \leftarrow \widetilde{V}_i^{1:N}$, $\widetilde{E}^{1:N} \leftarrow \widetilde{B}_i^{1:N}$, $\widetilde{G}^{1:N} \leftarrow (\widetilde{U}_i^{1:N}Z_i^{1:N})$, and $\delta_N^{(FG)} \leftarrow \delta_N^{(P)}$ by Lemma \ref{lem_dist_A}, we have by Lemma \ref{lemdrev3}
\begin{multline}
\mathbb{V}(\widetilde{p}_{B_i^{1:N}[\mathcal{V}_{V|UZ}] U_i^{1:N} Z_i^{1:N}},\widetilde{p}_{B_i^{1:N}[\mathcal{V}_{V|UZ}]} \widetilde{p}_{U_i^{1:N} Z_i^{1:N}})\\ {\leq}  \sqrt{2\log 2} \sqrt{N \delta_N} + 3\delta_N^{(P)}, \label{eq_sec_int3}
\end{multline}

Then, for $N$ large enough by \cite{bookCsizar},
\begin{align*}
& I(S_i \Psi^{V|U}_{i-1} ; Z_i^{1:N} \Phi^{U}_i \Psi^{U}_i)  \\
& \leq  I (\widetilde{B}_i^{1:N}[ \mathcal{V}_{V|UZ}] ; Z_i^{1:N} \widetilde{U}_i^{1:N})\\
& \leq   \mathbb{V}(\widetilde{p}_{B_i^{1:N}[\mathcal{V}_{V|UZ}] U_i^{1:N} Z_i^{1:N}},\widetilde{p}_{B_i^{1:N}[\mathcal{V}_{V|UZ}]} \widetilde{p}_{U_i^{1:N} Z_i^{1:N}}) \\
& \phantom{lm} \times \log_2 \frac{| \mathcal{V}_{V|UZ}|}{\mathbb{V}(\widetilde{p}_{B_i^{1:N}[\mathcal{V}_{V|UZ}] U_i^{1:N} Z_i^{1:N}},\widetilde{p}_{B_i^{1:N}[\mathcal{V}_{V|UZ}]} \widetilde{p}_{ U_i^{1:N} Z_i^{1:N}})} \\
& \leq  \sqrt{2\log 2} \sqrt{N \delta_N}(1+ 6 \sqrt{ 2 } + 3 \sqrt{3}) ( N   \\
& \phantom{lm}- \log_2 (\sqrt{2\log 2} \sqrt{N \delta_N}(1+ 6 \sqrt{ 2 } + 3 \sqrt{3})) ),
\end{align*}
where we have used (\ref{eq_sec_int3}) and that $x \mapsto x \log x$ is decreasing for $x>0$ small enough.
\section{Proof of Lemma~\ref{lem4}} \label{App_lem4}

With the substitution $F^{1:N} \leftarrow X^{1:N}$, $E^{1:N} \leftarrow T^{1:N}$, $G^{1:N} \leftarrow (U^{1:N}V^{1:N}Z^{1:N})$, $\widetilde{F}^{1:N} \leftarrow \widetilde{X}_i^{1:N}$, $\widetilde{E}^{1:N} \leftarrow \widetilde{T}_i^{1:N}$, $\widetilde{G}^{1:N} \leftarrow (\widetilde{U}_i^{1:N}\widetilde{V}_i^{1:N}Z_i^{1:N})$, and $\delta_N^{(FG)} \leftarrow \delta_N^{(P)}$ by Lemma \ref{lem_dist_A}, we have by Lemma \ref{lemdrev3}
\begin{align*}
&\mathbb{V}(\widetilde{p}_{T_i^{1:N}[\mathcal{V}_{X|UVZ}] U_i^{1:N} V_i^{1:N}Z_i^{1:N}},\widetilde{p}_{T_i^{1:N}[\mathcal{V}_{X|UVZ}]} \widetilde{p}_{  U_i^{1:N} V_i^{1:N} Z_i^{1:N}}) \\
& \leq \sqrt{2\log 2} \sqrt{N \delta_N}+ 3\delta_N^{(P)}. 
\end{align*}
Hence, since $\mathcal{V}_{X|VZ} = \mathcal{V}_{X|UVZ}$ by the Markov chain $U - V - X - Z$, we have
\begin{align}
\mathbb{V}^* \leq \sqrt{2\log 2} \sqrt{N \delta_N}+ 3\delta_N^{(P)}, \label{eq_sec_int3b}
\end{align}
where we have defined 
\begin{align*}
&\mathbb{V}^* \triangleq \\ 
& \phantom{ml}\mathbb{V}(\widetilde{p}_{T_i^{1:N}[\mathcal{V}_{X|VZ}] U_i^{1:N} V_i^{1:N}Z_i^{1:N}},\widetilde{p}_{T_i^{1:N}[\mathcal{V}_{X|VZ}]} \widetilde{p}_{ U_i^{1:N} V_i^{1:N} Z_i^{1:N}}).
\end{align*}

Then, for $N$ large enough,
\begin{align*}
& I( \Psi^{X|V}_i ; Z_{i}^{1:N} \Psi^{V|U}_{i-1} S_{i} \Phi^{U}_i \Psi^{U}_i) \\
& = I(\widetilde{T}_i^{1:N}[ \mathcal{V}_{X|VZ}] ; Z_{i}^{1:N} \widetilde{B}_{i}^{1:N}[ \mathcal{H}_{V|UZ}] \Phi^{U}_i \Psi^{U}_i )\\
& \leq I(\widetilde{T}_i^{1:N}[ \mathcal{V}_{X|VZ}] ; Z_{i}^{1:N} \widetilde{B}_{i}^{1:N} \widetilde{U}_{i}^{1:N})\\
&  \stackrel{(a)}{=} I(\widetilde{T}_i^{1:N}[ \mathcal{V}_{X|VZ}] ; Z_{i}^{1:N} \widetilde{V}_{i}^{1:N} \widetilde{U}_{i}^{1:N})\\
& \stackrel{(b)}{\leq}   \mathbb{V}^*\log_2 \frac{|\mathcal{V}_{X|VZ}|}{\mathbb{V}^*} \\
& \stackrel{(c)}{\leq} \sqrt{2\log 2} \sqrt{N \delta_N}(1+ 6 \sqrt{ 2 } + 3 \sqrt{3}) ( N \\
& \phantom{mm} - \log_2 (\sqrt{2\log 2} \sqrt{N \delta_N}(1+ 6 \sqrt{ 2 } + 3 \sqrt{3})) ),
\end{align*}
where $(a)$ holds by invertibility of $G_n$, $(b)$ holds by \cite{bookCsizar}, $(c)$ holds (\ref{eq_sec_int3b}) and because $x \mapsto x \log x$ is decreasing for $x>0$ small enough.

\section{Proof of Lemma~\ref{lemdifc}} \label{App_lemdifc}
Let $i \in \llbracket 1 , k-1 \rrbracket$. We have
\begin{align*}
& \widetilde{L}_{i+1}  - \widetilde{L}_{i} \\
& = I(S_{1:k}; \Psi^{U}_1 \Phi^{U}_{1:i+1} Z_{1:i+1}^{1:N}) - I(S_{1:k}; \Psi^{U}_1 \Phi^{U}_{1:i} Z_{1:i}^{1:N})\\
& = I(S_{1:k}; \Phi^{U}_{i+1} Z_{i+1}^{1:N} | \Psi^{U}_1  \Phi^{U}_{1:i} Z_{1:i}^{1:N}) \\
& = I(S_{1:i+1}; \Phi^{U}_{i+1} Z_{i+1}^{1:N} |\Psi^{U}_1 \Phi^{U}_{1:i} Z_{1:i}^{1:N})\\
& \phantom{mmmmmm} +I(S_{i+2:k}; \Phi^{U}_{i+1} Z_{i+1}^{1:N} | \Psi^{U}_1 \Phi^{U}_{1:i} Z_{1:i}^{1:N}  S_{1:i+1})  \\
&  \stackrel{(a)}{\leq}  I(S_{1:i+1} \Phi^{U}_{1:i} Z_{1:i}^{1:N}  ; \Phi^{U}_{i+1} Z_{i+1}^{1:N} |\Psi^{U}_1  )  \\
& \phantom{mmmmmm} +I(S_{i+2:k}; \Phi^{U}_{1:i+1}   Z_{1:i+1}^{1:N} S_{1:i+1}  \Psi^{U}_1   ) \\
& \stackrel{(b)}{=}  I(S_{1:i+1} \Phi^{U}_{1:i} Z_{1:i}^{1:N}  ; \Phi^{U}_{i+1} Z_{i+1}^{1:N} |\Psi^{U}_1  )  \\
& = I(S_{i+1}  ; \Phi^{U}_{i+1} Z_{i+1}^{1:N} |\Psi^{U}_1  ) \\
& \phantom{mmmmmm}+ I(S_{1:i} \Phi^{U}_{1:i} Z_{1:i}^{1:N} ; \Phi^{U}_{i+1} Z_{i+1}^{1:N} |\Psi^{U}_1 S_{i+1} )  \\
& \stackrel{(c)}{\leq} \delta_N^{(*)} + I(S_{1:i} \Phi^{U}_{1:i} Z_{1:i}^{1:N} ; \Phi^{U}_{i+1} Z_{i+1}^{1:N} |\Psi^{U}_1 S_{i+1} )   \\
& \leq \delta_N^{(*)} + I(S_{1:i} \Phi^{U}_{1:i} Z_{1:i}^{1:N} ; \Phi^{U}_{i+1} Z_{i+1}^{1:N} S_{i+1}|\Psi^{U}_1 )    \\
& \stackrel{(d)}{\leq} \delta_N^{(*)} + I(S_{1:i} \Phi^{U}_{1:i} Z_{1:i}^{1:N} \Psi^{V|U}_i \Psi^{X|V}_i ; \Phi^{U}_{i+1} Z_{i+1}^{1:N} S_{i+1} |\Psi^{U}_1 )  \\
& = \delta_N^{(*)} + I( \Psi^{V|U}_i\Psi^{X|V}_i; \Phi^{U}_{i+1} Z_{i+1}^{1:N} S_{i+1} |\Psi^{U}_1 ) \\
& \phantom{mmlm}+ I(S_{1:i} \Phi^{U}_{1:i} Z_{1:i}^{1:N}; \Phi^{U}_{i+1} Z_{i+1}^{1:N} S_{i+1} |\Psi^{V|U}_i\Psi^{X|V}_i\Psi^{U}_1)  \\
&  \stackrel{(e)}{=} \delta_N^{(*)} + I( \Psi^{V|U}_i\Psi^{X|V}_i; \Phi^{U}_{i+1} Z_{i+1}^{1:N} S_{i+1} |\Psi^{U}_1 )   \\
& \leq \delta_N^{(*)} + I( \Psi^{V|U}_i\Psi^{X|V}_i \Psi^{U}_1;  S_{i+1} )  \\
& \phantom{mmmmmm}+ I( \Psi^{V|U}_i\Psi^{X|V}_i ; \Phi^{U}_{i+1} Z_{i+1}^{1:N} | \Psi^{U}_1 S_{i+1} )   \\
& \stackrel{(f)}{=} \delta_N^{(*)} +  I( \Psi^{V|U}_i\Psi^{X|V}_i ; \Phi^{U}_{i+1} Z_{i+1}^{1:N} | \Psi^{U}_1 S_{i+1} )   \\
%
%
%
& = \delta_N^{(*)} +   I( \Psi^{V|U}_i  ; \Phi^{U}_{i+1} Z_{i+1}^{1:N} | \Psi^{U}_1 S_{i+1} ) \\
& \phantom{mmmmmm}+ I(  \Psi^{X|V}_i ; \Phi^{U}_{i+1} Z_{i+1}^{1:N} | \Psi^{V|U}_i \Psi^{U}_1 S_{i+1} ) \\
& \leq  \delta_N^{(*)} +   I( \Psi^{V|U}_i S_{i+1} ; \Phi^{U}_{i+1} Z_{i+1}^{1:N} \Psi^{U}_1 ) \\
& \phantom{mmmmmm}+ I(  \Psi^{X|V}_i ; \Phi^{U}_{i+1} Z_{i+1}^{1:N}  \Psi^{V|U}_i \Psi^{U}_1 S_{i+1} ) \\
& \stackrel{(g)}{\leq} 3\delta_N^{(*)},
\end{align*}
where $(a)$ holds by the chain rule and positivity of mutual information, $(b)$ holds by independence of $S_{i+2:k}$ with all the random variables of the previous blocks, $(c)$ holds by Lemma~\ref{lem1c} because $I(S_{i+1} ; \Phi^{U}_{i+1} Z_{i+1}^{1:N} |\Psi^{U}_1 ) \leq I(S_{i+1} ; \Phi^{U}_{i+1} Z_{i+1}^{1:N} \Psi^{U}_1 )$, in $(d)$ we introduce the random variable $\Psi^{V|U}_i$ and $\Psi^{X|V}_i$ to be able to break the dependencies between the random variables of Block $(i+1)$ and the random variables of the previous blocks, $(e)$ holds because $S_{1:i}\Phi^{U}_{1:i}Z_{1:i}^{1:N} \rightarrow \Psi^{V|U}_i \Psi^{X|V}_i \Psi^{U}_1 \rightarrow \Phi^{U}_{i+1} Z_{i+1}^{1:N} S_{i+1}$ (see Figure \ref{figFGD2}), $(f)$ holds because $(\Psi^{V|U}_i,\Psi^{X|V}_i,\Psi^{U}_i)$ is independent of $S_{i+1}$, $(g)$ holds by Lemmas~\ref{lem1c}, \ref{lem4} and because $\Psi^{X|V}_i$ is equal to $\Psi^{X|V}_1$.%
\section{Proof of Lemma \ref{lemcard_2}} \label{App_card}

Consider a source $(\mathcal{X} \mathcal{Y}, p_{XY})$ with $|\mathcal{X}| =q$, $q$ prime and  $\mathcal{Y}$ a countable alphabet.  Let $(X^{1:N},Y^{1:N})$ be $N$ i.i.d. realizations of this source, where $N \triangleq 2^n$, $n \in \mathbb{N}$. In the following, let $\oplus$ denote the modulo-$q$ addition. We start with some definitions and recall some useful results for our proof.

For a source $(\mathcal{X} \mathcal{Y}, p_{XY})$ the Bhattacharyya source parameter is defined by \cite{Sasoglu11}
$$
Z_s(W) \triangleq \frac{1}{q-1} \sum_{d \in \mathcal{X}\backslash\{0 \}} \sum_{ x \in \mathcal{X}} \sum_{y \in \mathcal{Y}} \sqrt{ p(x,y) p(x \oplus d,y)}.
$$

For a channel $W \triangleq (\mathcal{X}, W_{Y|X}, \mathcal{Y})$, the Bhattacharyya channel parameter is defined by \cite{Sasoglu09} 
$$
Z_c(W) \triangleq \frac{1}{q(q-1)} \sum_{d \in \mathcal{X}\backslash\{0 \}} \sum_{ x \in \mathcal{X}} \sum_{y \in \mathcal{Y}} \sqrt{ W(y|x)W(y|x \oplus d)}.
$$

Recall the following relations between Bhattacharyya parameters and corresponding source entropy and symmetric capacity.

\begin{prop}[\!\!{\cite[Prop. 3.3]{Sasoglu11},\cite[Prop. 3]{Sasoglu09}}]~\\\label{propBH} In this proposition, the base of the logarithm is chosen as $q = |\mathcal{X}|$. 
\begin{itemize}
\item For a source $(\mathcal{X} \mathcal{Y}, p_{XY})$, we have
\begin{align*}
 H(X|Y) & \geq Z_s(X|Y)^2.
\end{align*}
\item For a channel $W \triangleq (\mathcal{X}, W_{Y|X}, \mathcal{Y})$, we have
\begin{align*}
 I(W) & \geq \log \frac{q}{1 + (q-1) Z_c(W)},
\end{align*}
where 
$$
I(W) \triangleq \sum_{x \in \mathcal{X}} \sum_{y \in \mathcal{Y}} \frac{1}{q} W(y|x) \log \frac{ W(y|x)}{ \sum_{x' \in \mathcal{X}} \frac{1}{q} W(y|x')}
$$
denotes the symmetric capacity of the channel $W$.
\end{itemize}
\end{prop}

We have the following equivalence between the Bhattacharyya source parameter and the Bhattacharyya channel parameter. It is an extension of \cite[Th.2]{Honda13} to the $q$-ary case.

\begin{prop} \label{propeq}
Consider a source $(\mathcal{X} \mathcal{Y}, p_{XY})$ with $|\mathcal{X}| =q$, and  $\mathcal{Y}$ a countable alphabet. Let $ \widetilde{\mathcal{Y}} \triangleq \mathcal{X} \times \mathcal{Y}$, and 
\begin{align*}
	\widetilde{Y}^{1:N} \triangleq ( Z^{1:N} , Y^{1:N}) \text{ with } Z^{1:N} \triangleq \widetilde{X}^{1:N} \oplus X^{1:N},
\end{align*}
 where  $\widetilde{X}^{1:N}$ is uniformly distributed and independent of $(X^{1:N},Y^{1:N})$. Define $\widetilde{U}^{1:N} \triangleq \widetilde{X}^{1:N} G_n$, ${U}^{1:N} \triangleq {X}^{1:N} G_n$, and 
	$$
	\widetilde{W}_i^N( \widetilde{u}^{1:i-1}, \widetilde{y}^{1:N} | \widetilde{u}^i) \triangleq p_{ \widetilde{U}^{1:i-1} \widetilde{Y}^{1:N} | \widetilde{U}^i}( \widetilde{u}^{1:i-1}, \widetilde{y}^{1:N} | \widetilde{u}^i).
	$$ 
Then,	we have
	$$
	Z_s(U^i|U^{1:i-1}Y^{1:N}) = Z_c (\widetilde{W}_i^N ).
	$$
\end{prop}

\begin{proof}
Similar to \cite{Honda13}, we have
	\begin{align}
		&\widetilde{W}_i^N( \widetilde{u}^{1:i-1}, \widetilde{y}^{1:N} | \widetilde{u}^i) \nonumber \\ \nonumber
		& = p_{ \widetilde{U}^{1:i-1} \widetilde{Y}^{1:N} | \widetilde{U}^i}( \widetilde{u}^{1:i-1} \widetilde{y}^{1:N} | \widetilde{u}^i)\\ \nonumber
		& = \sum_{x^{1:N}} p_{ X^{1:N} Y^{1:N} \widetilde{X}^{1:N} \widetilde{U}^{1:i-1} | \widetilde{U}^i}(x^{1:N},y^{1:N}, \\
		& \phantom{mmmmmmmmmmmlmmm}z^{1:N}\oplus x^{1:N}, \widetilde{u}^{1:i-1}  | \widetilde{u}^i) \nonumber\\ \nonumber
		& \stackrel{(a)}{=} \sum_{x^{1:N}} p_{ X^{1:N} Y^{1:N}}(x^{1:N},y^{1:N})\\
		& \phantom{mmmmmm} \times p_{ \widetilde{X}^{1:N} \widetilde{U}^{1:i-1} | \widetilde{U}^i}(z^{1:N}\oplus x^{1:N}, \widetilde{u}^{1:i-1}  | \widetilde{u}^i) \nonumber \\ \nonumber
		& = \sum_{x^{1:N}} p_{ X^{1:N} Y^{1:N}}(x^{1:N},y^{1:N}) p_{ \widetilde{X}^{1:N}} (z^{1:N}\oplus x^{1:N}) \\ \nonumber
		& \phantom{mmmmmm} \times\frac{  p_{  \widetilde{U}^{1:i} | \widetilde{X}^{1:N}}(\widetilde{u}^{1:i}  | z^{1:N}\oplus x^{1:N})}{ p_{\widetilde{U}^i}(\widetilde{u}^i) }	\\ \nonumber
		& = \sum_{x^{1:N}} p_{ X^{1:N} Y^{1:N}}(x^{1:N},y^{1:N}) p_{ \widetilde{X}^{1:N}} (z^{1:N}\oplus x^{1:N})\\
		& \phantom{mmmmmm} \times \frac{ \mathds{1} \{  \widetilde{u}^{1:i} = ((z^{1:N} \oplus x^{1:N}) G_n )^{1:i}  \}}{ p_{\widetilde{U}^i}(\widetilde{u}^i) } \nonumber \\ \nonumber
		&\stackrel{(b)}{=} q^{-N+1}\sum_{x^{1:N}} p_{ X^{1:N} Y^{1:N}}(x^{1:N},y^{1:N}) \\
		& \phantom{mmmmmm} \times \nonumber\mathds{1} \{  \widetilde{u}^{1:i} \oplus (z^{1:N} G_n )^{1:i} = (x^{1:N} G_n )^{1:i}  \}  \\ 
		&= q^{-N+1} p_{ U^{1:i} Y^{1:N}}(\widetilde{u}^{1:i} \oplus (z^{1:N} G_n )^{1:i},y^{1:N}), \label{eqmodq}
			\end{align}
	where $(a)$ holds by independence of $(X^{1:N},Y^{1:N})$ and $(\widetilde{X}^{1:N},\widetilde{U}^{1:N})$, $(b)$ holds by uniformity of $\widetilde{X}^{1:N}$ and $\widetilde{U}^{1:N}$.
	
	We then have \eqref{eqspli},	
\begin{figure*}[]
\normalsize

\begin{align}
		& Z_c (\widetilde{W}_i^{N} ) \nonumber \\  \nonumber
		& =  \frac{1}{q(q-1)} \sum_{d \in \mathcal{X}\backslash\{0 \},  x ,\widetilde{y}^{1:N}, \widetilde{u}^{1:i-1}} \sqrt{ p_{ \widetilde{U}^{1:i-1} \widetilde{Y}^{1:N} | \widetilde{U}^i}( \widetilde{u}^{1:i-1} \widetilde{y}^{1:N} | x) p_{ \widetilde{U}^{1:i-1} \widetilde{Y}^{1:N} | \widetilde{U}^i}( \widetilde{u}^{1:i-1} \widetilde{y}^{1:N} | x \oplus d)}\\ \nonumber
		& \stackrel{(a)}{=}  \frac{q^{-N+1}}{q(q-1)} \sum_{d \in \mathcal{X}\backslash\{0 \} , x , y^{1:N},z^{1:N}, \widetilde{u}^{1:i-1}} \!\!\!\!\!\!\!\!\!\!\!\!\!\!\!\sqrt{ p_{ U^{1:i} Y^{1:N}}( (\widetilde{u}^{1:i-1},x) \! \oplus \! (z^{1:N} G_n )^{1:i},y^{1:N}) p_{ U^{1:i} Y^{1:N}}( (\widetilde{u}^{1:i-1}, x\oplus d) \! \oplus \! (z^{1:N} G_n )^{1:i},y^{1:N})}\\ \nonumber
		& \stackrel{(b)}{=} \frac{1}{q-1} \sum_{d \in \mathcal{X}\backslash\{0 \} , x ,y^{1:N},\widetilde{u}^{1:i-1}} \sqrt{ p_{ U^{1:i} Y^{1:N}}( (\widetilde{u}^{1:i-1} ,x),y^{1:N}) p_{ U^{1:i} Y^{1:N}}( (\widetilde{u}^{1:i-1} , x \oplus d),y^{1:N})}\\
		&\stackrel{(c)}{=} Z_s(U^i|U^{1:i-1}Y^{1:N}), \label{eqspli}
	\end{align}

\hrulefill
\vspace*{4pt}
\end{figure*}
	where $(a)$ holds by \eqref{eqmodq}, $(b)$ holds by doing the changes of variables $x \leftarrow x \oplus (z^{1:N} G_n )^{i}$ and $\widetilde{u}^{1:i-1} \leftarrow \widetilde{u}^{1:i-1} \! \oplus \! (z^{1:N} G_n )^{1:i-1}$, $(c)$ holds by definition of the Bhattacharyya source parameter.
\end{proof}

Recall also that for $q$-ary input symmetric channels, with $q$ prime, we have the following result.

\begin{prop} [\cite{karzand2010polar}] \label{propq}
For a $q$-ary input symmetric channel $W \triangleq (\mathcal{X}, p_{Y|X}, \mathcal{Y})$ with $q$-prime, define ${U}^{1:N} \triangleq {X}^{1:N} G_n$, where $X^{1:N}$ is uniformly distributed, and 
	$$
	{W}_i^N( {u}^{1:i-1}, {y}^{1:N} | {u}^i) \triangleq p_{ {U}^{1:i-1}, {Y}^{1:N} | {U}^i}( {u}^{1:i-1}, {y}^{1:N} | {u}^i).
	$$ 
	Define the symmetric capacity of $W_i^N$ by $I(W_i^N)$. Then, for $\delta_N \triangleq 2^{-N^{\beta}}$, $\beta < 1/2$, we have 
	$$
	\lim_{N\to \infty} \frac{| \{i \in \llbracket 1,N \rrbracket :  I(W_i^N) < \delta_N \}|}{N} = \log_2 (q)-I(W).
	$$
	
\end{prop}

We are now equipped to prove Lemma \ref{lemcard_2}. 	Let $\beta < 1/2$ and $\alpha < \beta$. Consider a source $(\mathcal{X} \mathcal{Y}, p_{XY})$ with $|\mathcal{X}| =q$, $q$ prime and  $\mathcal{Y}$ a countable alphabet. Let $ \widetilde{\mathcal{Y}} \triangleq \mathcal{X} \times \mathcal{Y}$, and 
\begin{align*}
	\widetilde{Y}^{1:N} \triangleq ( Z^{1:N} , Y^{1:N}) \text{ with } Z^{1:N} \triangleq \widetilde{X}^{1:N} \oplus X^{1:N},
\end{align*}
 where $\widetilde{X}^{1:N}$ is uniformly distributed and independent of $(X^{1:N},Y^{1:N})$. Define $\widetilde{U}^{1:N} \triangleq \widetilde{X}^{1:N} G_n$, ${U}^{1:N} \triangleq {X}^{1:N} G_n$, and 
	$$
	\widetilde{W}_i^N( \widetilde{u}^{1:i-1}, \widetilde{y}^{1:N} | \widetilde{u}^i) \triangleq p_{ \widetilde{U}^{1:i-1}, \widetilde{Y}^{1:N} | \widetilde{U}^i}( \widetilde{u}^{1:i-1}, \widetilde{y}^{1:N} | \widetilde{u}^i).
	$$ 
	We define
	$$
	\mathcal{A} \triangleq \{i \in \llbracket 1,N \rrbracket :  I(\widetilde{W}_i^N) < 2^{-N^{\beta}} \}
	$$ 
	and
	$$
	\mathcal{B} \triangleq \{i \in \llbracket 1,N \rrbracket :  H(U^i|U^{1:i-1}Y^{1:N}) > \log_2 (q)- 2^{-N^{\alpha}} \}.
	$$ 
	Assume $i \in \mathcal{A}$, then
	\begin{align*}
		&H(U^i|U^{1:i-1}Y^{1:N})\\ 
		& \stackrel{(a)}{\geq} \log_2 (q) Z_s(U^i|U^{1:i-1}Y^{1:N})^2\\
		& \stackrel{(b)}{=} \log_2 (q)Z_c(\widetilde{W}_i^N )^2 \\
		& \stackrel{(c)}{\geq} \log_2 (q) \left(\frac{ qe^{-2^{-N^{\beta}} \log(2)}-1} {q-1}\right)^2 \\
		& \stackrel{(d)}{\geq}\log_2 (q) \left(\frac{q (1-2^{-N^{\beta}} \log(2))-1} {q-1}\right)^2  \displaybreak[0] \\
		& = \log_2 (q) \left( 1 - 2^{-N^{\beta}}\frac{q\log(2)} {q-1} \right)^2 \displaybreak[0] \\
				& \geq  \log_2 (q) - 2^{-N^{\beta}}\frac{2q\log(2)} {q-1}  \\
								& \stackrel{(e)}{\geq}  \log_2 (q)- 2^{-N^{\alpha}},
	\end{align*}
	where $(a)$ holds by Proposition~\ref{propBH}, $(b)$ holds by Proposition~\ref{propeq}, $(c)$ holds because $i \in \mathcal{A}$ and by Proposition~\ref{propBH}, $(d)$ holds because $e^x \geq 1+x$, and $(e)$ holds for $N$ large enough because $\alpha<\beta$.
	Hence, for $N$ large enough, we have 
	$$
	\mathcal{A} \subseteq \mathcal{B},
	$$
	and thus by Proposition \ref{propq} and because $I(\widetilde{W})=\log_2 (q)- H(X|Y)$, we have
	\begin{align} \label{eqp1}
	   H(X|Y) =   \lim_{N\to \infty} \frac{| \mathcal{A}|}{N}\leq \lim_{N\to \infty} \frac{| \mathcal{B}|}{N}. 
	\end{align}
	
	Moreover, 
	$$
	\mathcal{B} \subseteq \{i \in \llbracket 1,N \rrbracket :  H(U^i|U^{1:i-1}Y^{1:N}) > 2^{-N^{\alpha}} \},
	$$
	 and we know by \cite{Sasoglu11} \begin{multline*}
	\lim_{N\to \infty} \frac{|  \{i \in \llbracket 1,N \rrbracket :  H(U^i|U^{1:i-1}Y^{1:N}) > 2^{-N^{\alpha}} \}|}{N}\\ = H(X|Y),
\end{multline*}
	which gives
	\begin{align} \label{eqp2}
	   H(X|Y) \geq   \lim_{N\to \infty} \frac{| \mathcal{B}|}{N}. 
	\end{align}
The combination of \eqref{eqp1} and \eqref{eqp2} proves the lemma.

\bibliographystyle{IEEEtran}
\bibliography{polarwiretap}

\end{document}